\documentclass[fleqn]{elsarticle}
\usepackage[mathscr]{eucal}
\usepackage{amsmath,amsfonts,amssymb,amsthm}
\usepackage{graphicx}
\usepackage{hyperref}
\usepackage{mathtools}
\usepackage[dvipsnames]{xcolor}
\usepackage{booktabs}
\usepackage{tabularx}
\usepackage{multirow}
\usepackage[subnum]{cases}
\usepackage{color}
\usepackage{cases}
\usepackage{natbib}
\usepackage{tgpagella}
\usepackage{subfigure}
\usepackage{caption}
\usepackage[english]{babel}
\theoremstyle{plain}
\newtheorem{theorem}{Theorem}

\newtheorem{lemma}{Lemma}

\newtheorem{proposition}{Proposition}

\newtheoremstyle{note}{\topsep}{\topsep}{\slshape}{}{\scshape}{}{ }{}
\theoremstyle{note}

\biboptions{sort&compress}
\numberwithin{equation}{section}
\numberwithin{theorem}{section}
\numberwithin{definition}{section}

\numberwithin{proposition}{section}
\numberwithin{corollary}{section}
\numberwithin{remark}{section}

\mathtoolsset{mathic,centercolon}
\newcommand\scE{{\mathscr E}}

\newcommand\scM{{\mathscr M}}
\newcommand\scN{{\mathscr N}}

\newcommand\mvector{\boldsymbol}

\newcommand\vp{\mvector{p}}
\newcommand\vq{\mvector{q}}

\newcommand\vv{\mvector{v}}

\newcommand\vx{\mvector{x}}

\newcommand\vz{\mvector{z}}

\newcommand\vJ{\mvector{J}}

\newcommand\vX{\mvector{X}}
\newcommand\vY{\mvector{Y}}

\newcommand\vvarphi{\mvector{\varphi}}

\newcommand\field{\mathbb}

\newcommand\R{\field{R}}

\newcommand\C{\field{C}}
\newcommand\Z{\field{Z}}

\newcommand\N{\field{N}}

\newcommand\const{\operatorname{const}}

\newcommand\rmd{\mathrm{d}}

\newcommand\rmi{\mathrm{i}\mspace{1mu}}
\newcommand\rme{\mathrm{e}}

\newcommand\Dz{\frac{\mathrm{d}\phantom{z} }{ \mathrm{d}z}}

\usepackage{subfigure}
\begin{document}
	
\begin{frontmatter}
	\title{Dynamics and  non-integrability of the double spring pendulum}
	
	\author{Wojciech Szumi\'nski$^{1*}$}
	\ead{w.szuminski@if.uz.zgora.pl}
	\address{$^1$Institute of Physics, University of Zielona G\'ora, Licealna 9, PL-65-407, Zielona G\'ora, Poland}
	\author{Andrzej J. Maciejewski$^{2}$}
	\ead{a.maciejewski@ia.uz.zgora.pl}
	\address{$^2$Janusz Gil Institute of Astronomy, University of Zielona G\'ora, Licealna 9, PL-65-407, Zielona G\'ora, Poland}
	\begin{abstract}
	This paper investigates the dynamics and integrability of the double spring pendulum, which has great importance in studying nonlinear dynamics, chaos, and bifurcations. Being a Hamiltonian system with three degrees of freedom, its analysis presents a significant challenge. To gain insight into the system's dynamics, we employ various numerical methods, including Lyapunov exponents spectra, phase-parametric diagrams, and Poincaré cross-sections. The novelty of our work lies in the integration of these three numerical methods into one powerful tool. We provide a comprehensive understanding of the system's dynamics by identifying parameter values or initial conditions that lead to hyper-chaotic, chaotic, quasi-periodic, and periodic motion, which is a novel contribution in the context of Hamiltonian systems.
In the absence of gravitational potential, the system exhibits $S^1$ symmetry, and the presence of an additional first integral was identified using Lyapunov exponents diagrams. We demonstrate the effective utilization of Lyapunov exponents as a potential indicator of first integrals and integrable dynamics. 
The numerical analysis is complemented by an analytical proof regarding the non-integrability of the system. This proof relies on the analysis of properties of the differential Galois group of variational equations along specific solutions of the system. To facilitate this analysis, we utilized a newly developed extension of the Kovacic algorithm specifically designed for fourth-order differential equations. Overall, our study sheds light on the intricate dynamics and integrability of the double spring pendulum, offering new insights and methodologies for further research in this field.
	\end{abstract}
	
	\begin{keyword}
 Variable-length pendulum; Double spring pendulum; Ordinary differential
 equations; Chaos in Hamiltonian
 systems; Numerical simulation;  Integrability analysis  
 \paragraph{Declaration} The article has been published in~\cite{Szuminski:24::}, and the final version is available at: \textbf{\href{https://doi.org/10.1016/j.jsv.2024.118550}{https://doi.org/10.1016/j.jsv.2024.118550} }
 
\end{keyword}

\end{frontmatter}
\section{Introduction  and motivation}
	The dynamics of various types of pendulums have a long history in mechanics and physics. Since the time of Galileo Galilei, analysis of their complicated behavior is still in great scientific activity. It is caused by the simplicity of these systems and by many fundamental and not obvious phenomena exhibited by pendulums.  Indeed, one can find numerous papers, books, and video clips concerning their non-linear and chaotic dynamics. To exemplify it, let us mention the paradigm models such as the spring pendulum~\cite{Broucke:73::,Lee:97::,Maciejewski:04::c,Awrejcewicz:08::, Amer:23::}, the magnetic pendulum~\cite{WOJNA2018214,ZHANG2020115549,SKURATIVSKYI2022116710}, the double and the triple pendulums~\cite{Shinbrot:92::,Stachowiak:06::,mp:13::c,Stachowiak:15::,PUZYROV2022116699,Nigmatullin:14::,MR4191961,MR4412899,Puzyrov:22::,Dyk:24::}, the coupled pendulums~\cite{Huynh2010,Huynh2013,Elmandouh:16::,Szuminski:20::,Szuminski:23::} and the swinging Atwood machine~\cite{Tufillaro:84::,Tufillaro:85::,Tufillaro:90::,Szuminski:22::,Olejnik:23b::}. These models have been extensively studied by researchers both theoretically and practically~\cite{Levien:93::,Pujol:10::,Gomez:21::,ciezkowski:21::,Pilipchuk:22::,Chu:22::,Olejnik:23b::}. In fact, pendulum systems have many potential applications~\cite{Wojna:18::,Liu:19::,Sharghi:22::,Yang:22::}.  For instance, multiple pendulum models play a crucial role in engineering and in synchronization theory~\cite{Dilao:09::,KOLUDA2014977,DUDKOWSKI20181,PhysRevLett.72.2009}.  Analyzing the dynamics of the double pendulum system helps engineers design control algorithms, especially for bipedal robots where motion stability is a critical aspect~\cite{4153387,Sahin:17::}. 
However, a simple spring pendulum can be treated as a classical analog of the quantum phenomenon of Fermi resonance in the infrared spectrum of carbon dioxide~\cite{Vitt:33::,MR1751314}, and currently it has been treated as a system with potential applications in atmosphere modeling~\cite{MR1948160,MR2043791,DeShazer}. 

 
 Dynamics and integrability analysis of pendulum-like systems with three or more degrees of freedom have not yet been fully explored. This is for several reasons. Namely, 
for Hamiltonian systems with two degrees of freedom, an essential tool for giving a quick insight into the system dynamics is the so-called   Poincar\'e cross-sections method. This technique by a simple cross-sections of the phase curves with a two-dimensional plane shows the coexistence of periodic, quasi-periodic, and chaotic orbits at the section plane giving a qualitative insight into system dynamics. However, for Hamiltonian systems with more than two degrees of freedom, this method is less practical.  One significant issue arises from the fact that the intersections of trajectories with a Poincar\'e section hyperplane which is of dimension higher than two.  This complicates the extraction of meaningful information, as the resulting section may exhibit intricate and convoluted structures that are difficult to interpret. Moreover, higher dimensionality introduces additional complexities, such as the presence of more intricate bifurcation scenarios and a richer set of possible trajectories.

 On the other hand, computations of Lyapunov exponents provide a quantitative description of chaos and its strength and can be effectively applied to a system with many degrees of freedom.
 The analysis of the nontrivial exponent values provides valuable information, such as the presence of new isolating integrals.   However,  in a regular regime, where all exponents are zero, periodic and quasi-periodic solutions are not distinguishable. Therefore, in a recent article~\cite{Szuminski:23::}, the authors combined these two methods by providing compressive information on the dynamics of a model considered.

Nevertheless, the integrability analysis of pendulum-like systems is difficult due to their dependence on several parameters, such as lengths of pendulum arms, masses of bobs, spring stiffness, gravity acceleration, etc. Therefore, to perform a complete integrability analysis of these models, one needs a powerful tool, which enables one to distinguish values of parameters for which a considered system is suspected to be integrable. Such an effective and strong tool is the so-called Morales--Ramis theory~\cite{Morales:99::,Morales:00::}. Let us recall the main theorem of this theory.

\begin{theorem}[Morales--Ramis (1999)]
 If a Hamiltonian system is integrable in the sense of Liouville in a neighborhood of a particular solution, then the identity component of the
  differential Galois group of the variational equations along this solution is Abelian.
\end{theorem}	
		 The Morales-Ramis theory has already been successfully applied to various important physical systems~\cite{Yagasaki:18::,Acosta:18::,Acosta:18b::,Huang:18::,Combot:18::,Mnasri:18::,Shibayama:18::, Maciejewski:18::} as well as non-Hamiltonian systems~\cite{Huang:18::,Szuminski:18::,Maciejewski:20e::,Szuminski:20b::}, to cite just a few. Because of this, many new integrable and super-integrable systems were found~\cite{	Elmandouh:18::,Szuminski:18a::,Szuminski:18b::}. To exemplify it, we mention two generalizations of the swinging Atwood machine model recently studied in~\cite{Szuminski:22::,Szuminski:23::}. In this paper, the authors performed a detailed integrability analysis and found integrable and super-integrable cases with additional first integrals quadratic and quartic in the momenta. These first integrals were later used in constructions of general solutions of nonlinear equations of motions. 
		 
		 In most cases, however, the Morales-Ramis theory has been applied to Hamiltonian systems with two degrees of freedom for which a system of variational equations splits into two subsystems of linear equations. Next, each of these subsystems can be transformed (in general) into an equivalent second-order equation with rational coefficients. For such equations, there exists an algorithm called the Kovacic algorithm~\cite{Kovacic:86::}, which can be used to determine the differential Galois group properties of rationalized variational equations. For Hamiltonian systems with three degrees of freedom, the normal variational equations form a four-dimensional subsystem in most cases and the analysis of its differential Galois group is considerably more complicated.   
Unfortunately, there is no equivalent of the Kovacic algorithm for linear differential equations with rational coefficients of higher order, although many partial results are known~\cite{Singer:95::,Ulmer:03::}. In recent work~\cite{Combot:18b::}, the authors present an algorithm to study the differential Galois group for symplectic differential operators of dimension four, which is perhaps the most comprehensive study available. We will be using this algorithm in our considerations.

Currently, there is great activity studying variable-length pendulum systems, such as the swinging Atwood machine~\cite{Tufillaro:90::} and its generalizatins~~\cite{Szuminski:22::}, the variable-length coupled pendulums recently studied in~\cite{Szuminski:23::}, or the double variable-length pendulum with counterweight mass~\cite{Yakubu:21::,Yakubu:22::,Olejnik:23::}.  
Variable-length pendulum systems are excellent examples for studying nonlinear dynamics, chaos, and bifurcations. 
 Moreover, such models are interesting due to their potential physical applications in crane models and lifting equipment, where understanding motion and stability is crucial for safe and efficient operation~\cite{JU2006376,MR4459645,Freundlich:20::,Shahbazi:16::,Hayati:18::}, The flexibility and maneuverability of the variable length pendulum system make it important in robotics, where dynamic stability is crucial~\cite{PLAUT20133768,YANG2022116727,SHARGHI2022117036}.    Furthermore, the variable lengths pendulum systems have applications in energy conversion and storage, where swinging can be used to generate electricity~\cite{MARSZAL2017251,doi:10.1177/14613484221077474,ABOHAMER2023377}. 

Finally, studying the dynamics of variable length pendulum system can find its possible application in active debris removal missions~\cite{Shi:18::,Aslanov:24::,Bourabah:22::}.  One of the most promising models is a tethered tug-debris system. 
Usually, it consists of the mother satellite of mass $m_1$ moving in an unperturbed Kepler circular orbit and the sub-satellite of mass $m_2$ attached to the mother satellite through an elastic massless tether of length $l(t)$.  Although active debris removal using a tethered tug-debris system is a relatively new topic, several books~\cite{Levin:07::,Troger:10::,Aslanov:12::} and hundreds of scientific articles have been already published~\cite{Aslanov:12::,Aslanov:13::,Ledkov:19::,Shahbazzadeh:22::}. However, many aspects of this problem remain unexplored.  For example, the equations governing the dynamics of tethered satellites are highly nonlinear. Hence the dynamical behavior is very rich and in some cases can be chaotic~\cite{Misra:01::,Jin:16::,Aslanov:16::}.   Because chaos in dynamical systems in space can be exceedingly dangerous, especially concerning the stability and predictability of artificial celestial body movements, it is crucial to detect values of parameters and initial conditions for which the motion of the system will be regular and predictable.  Indeed,  small changes in initial conditions or parameters can lead to drastically different trajectories along which the system will move.

  The aim of this work is to perform a comprehensive analysis of the dynamics and integrability of the double-spring pendulum system, which can be regarded as a mathematical model of the tethered satellite system.
 Because the proposed model is a Hamiltonian system with three degrees of freedom, its analysis is quite challenging. The equations of motion obtained are strongly nonlinear and require careful numerical analysis to obtain reliable results. However, there is a lack of literature on its dynamics and integrability analysis. This is the main topic of this paper. 
 To study their complex dynamics, we employ numerical methods such as Lyapunov exponent diagrams and phase-parametric diagrams. Moreover, for special cases, we use the Poincar\'e cross-section method.  We join these three numerical techniques to get a complete insight into the dynamics of the considered models. For their exhaustive integrability analysis, we utilize the differential Galois theory and the Kovacic algorithm of dimension four.

 The rest of the paper is organized as follows. In Sec.~\ref{sec:model_1},  we introduce the model under consideration and derive the corresponding equations of motion.  In Sec.~\ref{sec:numerical}, we give a qualitative and quantitative description of its chaotic motion by joining numeral methods, such as Lyapunov’s exponent spectra, phase-parametric diagrams, and Poincar\'e cross-sections. Next, in Sec.~\ref{sec:integrability},  we perform an effective integrability analysis of the model with the help of Morales–Ramis theory and the application of the Kovacic algorithm of dimension four. The proof of its non-integrability for all values of the parameters is given. 
The basic facts and theorems devoted to the integrability analysis of variational equations of dimension four are included in the appendix.  
	\section{Description of the system \label{sec:model_1}}

	\begin{figure}[t]
\centering{
\resizebox{70mm}{!}{\begingroup%
		  \makeatletter%
		\providecommand\color[2][]{%
			\errmessage{(Inkscape) Color is used for the text in Inkscape, but the package 'color.sty' is not loaded}%
			\renewcommand\color[2][]{}%
		}%
		\providecommand\transparent[1]{%
			\errmessage{(Inkscape) Transparency is used (non-zero) for the text in Inkscape, but the package 'transparent.sty' is not loaded}%
			\renewcommand\transparent[1]{}%
		}%
		\providecommand\rotatebox[2]{#2}%
		\newcommand*\fsize{\dimexpr\f@size pt\relax}%
		\newcommand*\lineheight[1]{\fontsize{\fsize}{#1\fsize}\selectfont}%
		\ifx\svgwidth\undefined%
		\setlength{\unitlength}{248.83760708bp}%
		\ifx\svgscale\undefined%
		\relax%
		\else%
		\setlength{\unitlength}{\unitlength * \real{\svgscale}}%
		\fi%
		\else%
		\setlength{\unitlength}{\svgwidth}%
		\fi%
		\global\let\svgwidth\undefined%
		\global\let\svgscale\undefined%
		\makeatother%
  \begin{picture}(1,0.61995447)%
	        \put(0,0){\includegraphics[width=\unitlength,page=1]{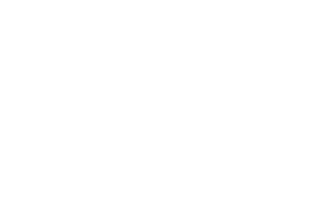}}%
    \put(0.3035064,0.17546961){\color[rgb]{0,0,0}\makebox(0,0)[lt]{ \smash{  \large{$m_1$} }}}%
    \put(0,0){\includegraphics[width=\unitlength,page=2]{Fig_1}}%
    \put(0.92636166,0.01347408){\color[rgb]{0,0,0}\makebox(0,0)[lt]{ \smash{  \large{$m_2$} }}}%
    \put(0.36861216,0.21843045){\color[rgb]{0,0,0}\makebox(0,0)[lt]{ \smash{   }}}%
    \put(0.13612637,0.44951251){\color[rgb]{0,0,0}\makebox(0,0)[lt]{ \smash{  \large{$\varphi_1$} }}}%
    \put(0.44236514,0.10692554){\color[rgb]{0,0,0}\makebox(0,0)[lt]{ \smash{  \large{$\varphi_2$} }}}%
    \put(0,0){\includegraphics[width=\unitlength,page=3]{Fig_1}}%
    \put(0.33073966,0.37193709){\color[rgb]{0,0,0}\makebox(0,0)[lt]{ \smash{  \large{$l_1$} }}}%
    \put(0.65695179,0.15667904){\color[rgb]{0,0,0}\makebox(0,0)[lt]{ \smash{  \large{$l_2(t)$} }}}%
    \put(0,0){\includegraphics[width=\unitlength,page=4]{Fig_1}}%
      \end{picture}
    	\endgroup}
\caption{Geometry of the double spring pendulum. The system oscillates under the influence of the gravitational and Hooke’s potentials. Here $m_1$ and $m_2$ are masses of the bobs, while $l_1$ and $l_2(t)$ are lengths of the pendulum's arms. The Lagrange function describing the model is defined in Eq.~\eqref{eq:lagrangian}}
\label{fig:geometria1}
}
\end{figure}
In Fig.~\ref{fig:geometria1} the geometry of the system is presented. We study the double spring pendulum system.  It consists of a simple mathematical pendulum of mass $m_1$ and length $l_1$  and a weightless spring of length $l_2=l_2(t)$ with a mass $m_2$ attached at its end. The system moves under the influence of the gravitational and Hooke's potential forces. 

The Lagrange function of the system is as follows
\begin{equation}
\begin{split}
\label{eq:LL}
L&=T-V_g-V_k,\\
T&=\frac{1}{2}m_1\left(\dot x_1^2+\dot y_1^2\right)+\frac{1}{2}m_2\left(\dot x_2^2+\dot y_2^2\right),\\
V_g&=-g m_1 x_1-g m_2 x_2,\\
V_k&=\frac{1}{2}k\left(\sqrt{\left(x_2-x_1\right)^2+\left(y_2-y_1\right)^2}-l_{20}\right)^2
\end{split}
\end{equation}
where  $k\in \R^+$ is the Young modulus of the spring and $l_{20}$ is its natural length.
Following Fig.~\ref{fig:geometria1}, we utilize the polar coordinates \begin{equation}
\label{eq:polar}
\begin{split}
x_1(t)&=l_1\cos\varphi_1(t),\\ y_1(t)&=l_1\sin\varphi_1(t),\\
x_2(t)&=x_1(t)+l_2(t)\cos\varphi_2(t),\\ y_2(t)&=y_1(t)+l_2(t)\sin\varphi_2(t).
\end{split}
\end{equation}
Hence, we obtain a system of three degrees of freedom defined by the Lagrange function
	\begin{equation}
		\begin{split}
			\label{eq:lagrangian}
			L&= 
			\frac{1}{2}\left(m_1+m_2\right)l_1^2\dot \varphi_1^2 +\frac{1}{2}m_2\left(\dot l_2^2+l_2^2\dot\varphi_2^2\right)+m_2 l_1 \dot\varphi_1\left[l_2\dot\varphi_2\cos(\varphi_1-\varphi_2)-\dot l_2\sin(\varphi_1-\varphi_2)\right],\\
		&+(m_1+m_2)g l_1\cos\varphi_1-m_2g l_2\cos\varphi_2-\frac{1}{2}k\left(l_2-l_{20}\right)^2.
		\end{split}
	\end{equation}
To minimize the number of parameters and thus simplify our calculations as much as possible, we rescale $
	\ell(t)=l_2(t)/l_1$, and we introduce a new time as 
	\begin{equation}
	\label{eq:time}t\to\omega_k^{-1} \,t,\qquad \text{where}\qquad \omega_k=\sqrt{\frac{k}{m_2}}.
	\end{equation}
 The dimensionless form of Lagrangian~\eqref{eq:lagrangian} is as follows.
		\begin{equation}
	\label{eq:lag_{resc}}
	\begin{split}
		L&=\frac{1}{2}\left(\dot \ell^2+\ell^2\dot\varphi_2^2\right)+\frac{1}{2}\left(\mu+1\right)\dot\varphi_1^2+\ell\cos(\varphi_1-\varphi_2) \dot\varphi_1\dot\varphi_2-\sin(\varphi_1-\varphi_2)\dot\ell\dot\varphi_1\\ &+\omega[(\mu+1)\cos\varphi_1+\ell \cos\varphi_2]- \frac{(\delta-\ell)^2}{2}.
		\end{split}
	\end{equation}
Here $\mu,\delta,\omega\in \R^+$ are dimensionless parameters defined as
	\begin{equation*}
	\label{eq:parek}
		\mu:= \frac{m_1}{m_2},\quad \delta:=\frac{l_{20}}{l_1},\quad \omega :=\frac{\omega_g^2}{\omega_k^2},\quad  \text{where}\quad \omega_g=\sqrt{\frac{g}{l_1}}.
	\end{equation*}
In our analysis, we will use only Lagrange variables. One can introduce canonical variables in order to use Hamilton's formalism directly; however, in canonical variables, the equations of motion are more complicated.    We will use the following variables: the
angular velocities $\omega_1=\dot \varphi_1, \omega_2=\dot \varphi_2$ and
$v=\dot \ell$ instead of momenta.  Hence, a set of three second-order Lagrange
equations can be rewritten as a system of six first-order differential equations
of the form

	\begin{equation} \begin{cases}
			\label{eq:vv}
			\begin{split}
				&\dot \ell=v,\\
				&   \dot v=\ell\omega_2^2+\cos(\varphi_1-\varphi_2)\omega_1^2+\omega\cos\varphi_1\cos(\varphi_1-\varphi_2) +\frac{\delta-\ell}{2\mu}\left(1+2\mu-\cos[2(\varphi_1-\varphi_2)]\right),\\
				&\dot\varphi_1=\omega_1, \\
				& \dot\omega_1=\frac{\delta-\ell}{\mu}\sin(\varphi_1-\varphi_2)-\omega\sin\varphi_1,\\
				&\dot\varphi_2=\omega_2,\\
				&\dot \omega_2=\frac{1}{\ell}\left(\sin(\varphi_1-\varphi_2)\omega_1^2-2v \omega_2\right)+\frac{1}{\ell}\left(\frac{\ell-\delta}{2}\sin[2(\varphi_1-\varphi_2)]+\omega\cos\varphi_1\cos(\varphi_1-\varphi_2)\right).
		\end{split}  \end{cases}
	\end{equation}
	System~\eqref{eq:vv} has the following energy first integral
	\begin{equation}
	\begin{split}
	\label{eq:E}
	E&=\frac{1}{2}\left(\dot \ell^2+\ell^2\dot\varphi_2^2\right)+\frac{1}{2}\left(\mu+1\right)\dot\varphi_1^2+\ell\cos(\varphi_1-\varphi_2) \dot\varphi_1\dot\varphi_2-\sin(\varphi_1-\varphi_2)\dot\ell\dot\varphi_1\\ &-\omega[(\mu+1)\cos\varphi_1+\ell \cos\varphi_2]+\frac{(\delta-\ell)^2}{2}.
	\end{split}
	\end{equation}
	\section{Numerical analysis \label{sec:numerical}}


	To visualize the dynamics of the system, we perform a numerical analysis with the help of Lyapunov's 	To visualize the dynamics of the system, we perform a numerical analysis with the help of Lyapunov's exponent diagrams (Lyapunov's diagrams in short), phase-parametric diagrams, and the Poincar\'e cross-sections method. In order, to get the most reliable results, we systematically join these three methods and we give a complete picture of the system's complex dynamics.
		\begin{figure}[t]
	\centering
		\includegraphics[width=0.45\linewidth]{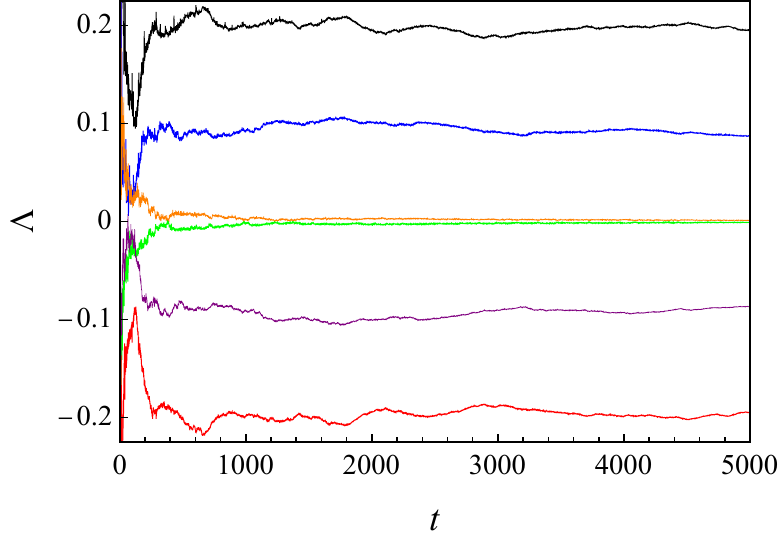}
	\caption{(Color online) The Lyapunov exponents spectrum $\Lambda=\{\lambda_1,\lambda_2,\lambda_3,\lambda_4,\lambda_5,\lambda_6\}$ of system~\eqref{eq:vv}, computed for the constant values of parameters~\eqref{eq:parki} and initial condition~\eqref{eq:ini}. 
For a sufficient amount of time steps, the convergence of the Lyapunov exponents is ensured. Because two Lyapunov exponents are positive, the system reveals hyperchaotic dynamics
	\label{fig:lyap_{time}}}
\end{figure}

		\begin{figure*}[t]
		\centering
		\includegraphics[width=0.37\linewidth]{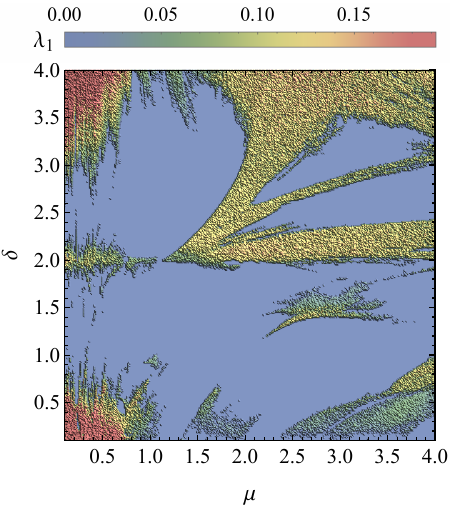}\hspace{1.5cm}
		\includegraphics[width=0.37\linewidth]{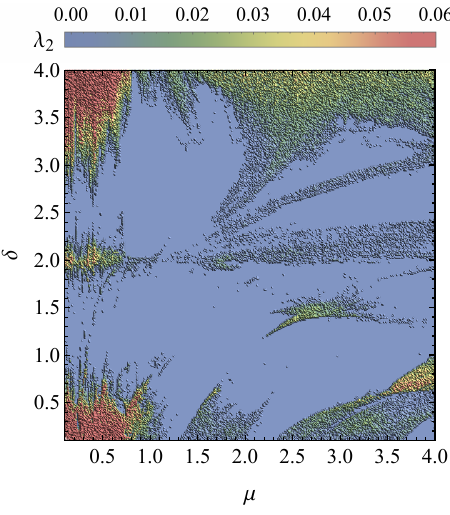}	
		\caption{(Color online)
The diagrams of two Lyapunov exponents  $\lambda_1$ and $\lambda_2$
computed for the grid of $400\times 400$ values of  $\mu,\delta\in[0.1,4]$  with $\omega=1$. The numerical integrations were performed successively for the initial condition~\eqref{eq:ini}.  	 Color scales are proportional to the magnitudes of $\lambda_1$ and $\lambda_2$, respectively.  The blue color represents the regular regions, while the remaining domain is responsible for the system's chaotic motion.  These two diagrams mostly coincide, which confirms the hyperchaotic nature of the system}
		\label{fig:lap_parki3}
	\end{figure*}
	\begin{figure*}[htp]
		\centering

				\includegraphics[width=0.37\linewidth]{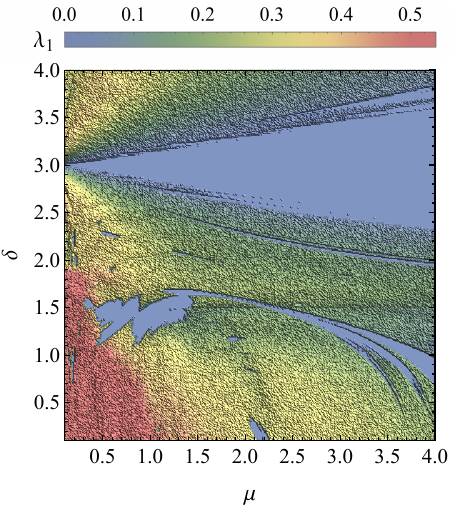}\hspace{1.5cm}
		\includegraphics[width=0.37\linewidth]{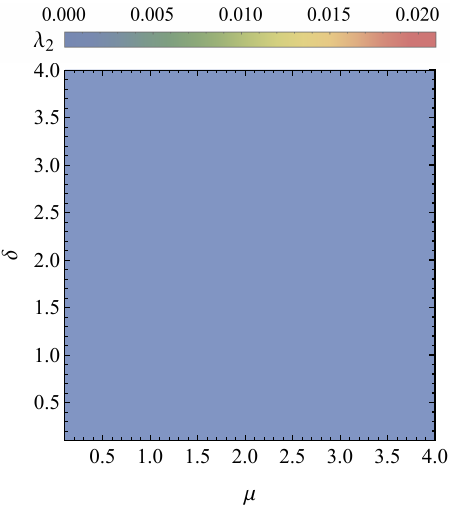}\\		\caption{(Color online) The diagrams of two Lyapunov exponents  $\lambda_1$ and $\lambda_2$
computed for the grid of $400\times 400$ values of   $\mu,\delta\in[0.1,4]$  with $\omega=0$. The numerical integrations were carried out successively for the initial condition~\eqref{eq:ini}.  	 The color scales are determined by the magnitudes of $\lambda_1$ and $\lambda_2$, respectively. Blue regions correspond to regular (non-chaotic) dynamics, while the remaining domain is responsible for complex dynamics. Due to the presence of the additional first integral, the exponent $\lambda_2=0$}
		\label{fig:lap_parki4}
	\end{figure*}
	\subsection{Diagrams of Lyapunov exponents}
	
	\begin{figure*}[htp]
	\centering
	\subfigure[\, \,$\omega=1$]{
	\includegraphics[width=0.43\linewidth]{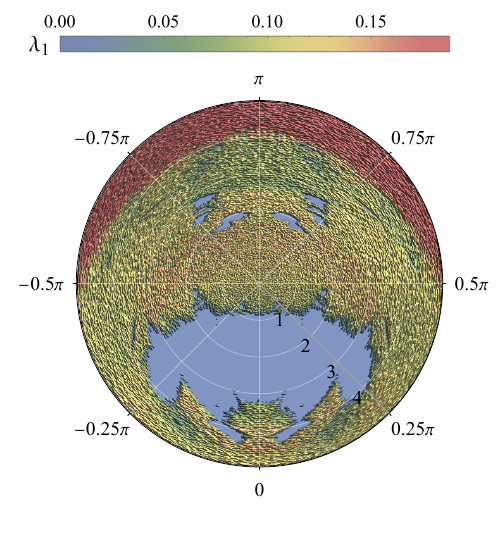}}	\hspace{1.3cm}
\subfigure[\, \,$\omega=1.6$]{
\includegraphics[width=0.43\linewidth]{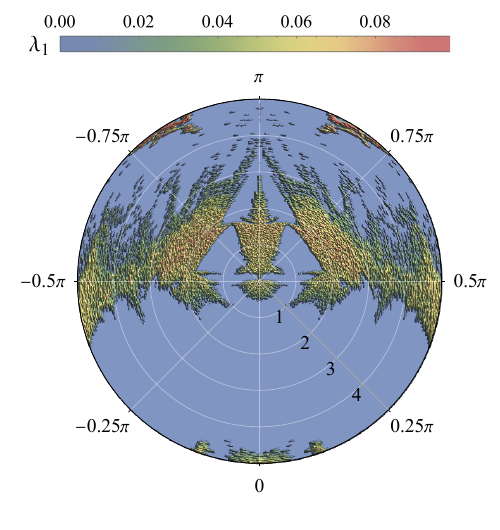}}\\
\subfigure[\, \,$\omega=6$]{
\includegraphics[width=0.43\linewidth]{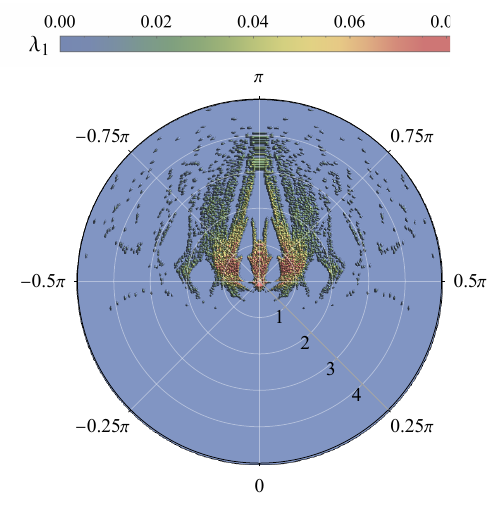}}\hspace{1.3cm}
\subfigure[\, \,$\omega=100$]{
\includegraphics[width=0.43\linewidth]{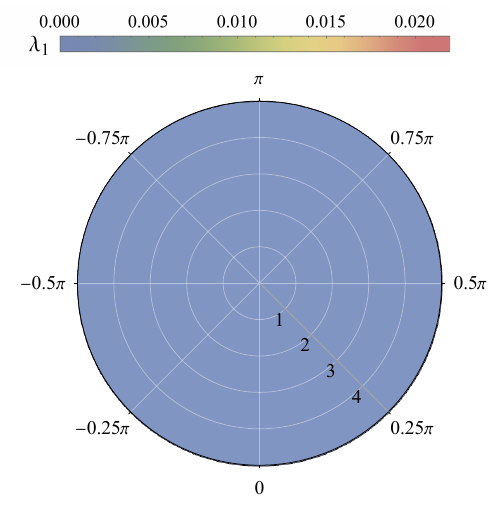}}
	\caption{(Color online)  The Lyapunov diagrams for the system~\eqref{eq:vv} in the polar plane  $(\ell(0),\varphi_2(0))$ constructed for the parameters  $\mu=3,\delta=1$ with varying $\omega$,  and initial conditions~\eqref{eq:ini_2}.
In the radial direction, we measure $\ell(0)\in[0.001,5]$; in angular direction, we measure $\varphi_2(0)\in[-\pi,\pi]$.  The color scale is proportional to the magnitude of the largest exponent $\lambda_1$. The plots visualize two zones: regular and chaotic. Blue regions indicate regular dynamics, while regions with $\lambda_1>0$ correspond to the system's chaotic behavior 	}\label{fig:lyap_polar1}
\end{figure*}
Lyapunov exponent is a measure of a system’s predictability and sensitivity to changes in its initial conditions.  Thus, it is commonly used in studying chaos in dynamical systems by 
quantifying the rate of exponential divergence of nearby trajectories in a phase space. According to chaos theory, if at least one Lyapunov exponent is positive, then a considered system
is sensitive to changes in its initial conditions and chaotic phenomena appear. 
Moreover, if two or more Lyapunov exponents are positive, then the system is treated as the hyperchaotic one~\cite{Sprott:10::}.  In our research, we use the standard algorithm introduced by Benettin et al.~\cite{Benettin:80::} to calculate two-parameter diagrams of the Lyapunov exponents. 	This technique is based on the integration of variational equations for $n$ initial conditions with successive applications of the Gram-Schmidt orthonormalization procedure.
In the presented study, we employ the standard algorithm implemented in Mathematica by Sandri~\cite{Sandri:96::}. However, in order to enhance the speed and precision of our results, we utilize the NDSolve solver with the "ExplicitRungeKutta" method  instead of Euler’s method.  After conducting a thorough and detailed analysis, we select a period between re-orthonormalizations of $T=1$, a maximum step size of $\tau=0.01$, and a number of steps typically set to $k=5000$ or more. The working precision for the entire numerical analysis is set to at least 12, ensuring the maintenance of precision up to 12 digits during internal computations. Moreover, the constancy of the energy first integral $E$, as given in~(2.7), is utilized to verify the numerical integrations. We maintain both relative and absolute errors within the range of~$10^{-11}$. These chosen parameters facilitate the efficient computation of Lyapunov exponents, ensuring their robust convergence and enabling possible fast calculations.

Fig.~\ref{fig:lyap_{time}} shows the Lyapunov exponents spectrum for system~\eqref{eq:vv}, computed for constant values of the parameters
\begin{equation}
\label{eq:parki}
\mu=3,\quad \delta=1,\quad \omega=1,
\end{equation}
under  the initial condition
\begin{equation}
\label{eq:ini}
\begin{split}
& \ell(0)=3,\quad \varphi_1(0)=0,\quad \varphi_2(0)=\frac{\pi}{4},\quad v(0)=\omega_1(0)=\omega_2(0)=0.
\end{split}
\end{equation}
As the considered system has a six-dimensional phase space,  there are six Lyapunov exponents~$\Lambda=\{\lambda_i\}$, $i=1,\ldots,6$, where $\lambda_1$ is the largest Lyapunov exponent. As we can notice, the integration time $~5000$ units was sufficient to ensure the convergence of the Lyapunov exponents. 
Moreover, we observe that they
 appear in additive inverse pairs, so they sum to zero. This is by Liouville's theorem~\cite{Liouville:83::}, which states that a conservative Hamiltonian system is volume-preserving. Therefore, if $\lambda$ is a Lyapunov exponent, then $-\lambda$ is also. Due to the existence of the first integral in the system~\eqref{eq:vv}, which is the conservation of total energy $E$, two Lyapunov exponents tend to zero.  The numerically computed values of $\lambda_3,\lambda_4$   are not exactly zero on finite time scales. Therefore, for practical purposes of our investigation, we take the zero value of the Lyapunov exponent whenever it is less than $0.003$.  
 Thus, the possible Lyapunov exponent spectrum of the system~\eqref{eq:vv} is $\Lambda=\{\lambda_1,\lambda_2,0,0,-\lambda_2,-\lambda_1\}$, where $ \lambda_1$ is the maximum Lyapunov exponent. As $\lambda_1$ and $\lambda_2$ are non-zero for the given values of the parameters~\eqref{eq:parki} and the initial conditions~\eqref{eq:ini}, the system dynamics is hyperchaotic.
 
By repeating the above procedure with varying parameter values or initial conditions, we can create Lyapunov diagrams on the parameter plane. 	Figs.~\ref{fig:lap_parki3}-\ref{fig:lap_parki4} present diagrams of Lyapunov exponents~$\lambda_1$ and $\lambda_2$ for the parameters pairs~$(\mu, \delta)$ with varying values of $\omega$. Color scales are proportional to the magnitudes of the exponents $\lambda_1$ and $\lambda_2$, respectively.  The colorful diagrams were obtained by numerical computations of Lyapunov exponents on a grid of $400\times 400$ values of parameters $\mu,\delta\in[0.1,4]$ with initial conditions~\eqref{eq:ini}. These diagrams show how changes in the values of the parameters $(\mu,\delta)$ affect the dynamics of the system. 
The blue region corresponds to regular, non-chaotic oscillations, and the rest of the domain is responsible for the chaotic motion with two positive Lyapunov exponents. As we can see in Fig.~\ref{fig:lap_parki3}, for $\omega=1$ the diagrams of exponents $\lambda_1$ and $\lambda_2$ mostly coincide confirming the hyperchaotic nature of the system. For relatively small values of the mass ratio $\mu\in[0.1,0.75]$ and for highly stretched $\delta\in [0.1,0.5]$, and compressed $\delta\in [3.5,4]$ spring, the strength of chaos is prominent with the largest Lyapunov exponent reaching its maximal value $\lambda_1\approx 0.19$. In the prescribed Lyapunov diagram, we can see that there are four large regular seas separated by chaotic regions. Inside these regular regions, one can find values of $\mu$ and $\delta$ for which the system's motion is periodic. We will explore this in more detail in the next subsection.   

The situation is completely different when we set $\omega=0$.  The diagram illustrated in Fig.~\ref{fig:lap_parki4}
shows complex dynamics for a wide range of values $(\mu, \delta)$, while the diagram of $\lambda_2$ presents a regular pattern with $\lambda_2\approx 0$ for every $(\mu, \delta)$. 
 It is caused by the presence of the additional first integral inside the system. Indeed, for $\omega=0$ the system possesses symmetry $\field{S}^1$ since Lagrangian~\eqref{eq:lag_{resc}} depends only on angle differences. Therefore, for $\omega=0$ the system can be reduced to a model of two degrees of freedom for which the Poincar\'e sections method can be adopted.  The above makes the Lyapunov exponents spectrum a possible indicator for searching for additional first integrals and integrable dynamics. Indeed, in a recent paper~\cite{Szuminski:23::}, the author algorithmically used the method of Lyapunov exponents in the systematic search for the first integrals of the systems. When examining the bottom left diagram presented in Fig.~\ref{fig:lap_parki4}, we can observe some kind of tunnel that leads to the regular sea through a chaotic area. As we shall discover, many periodic solutions exist within this region.

Lyapunov exponent diagrams can also be useful for understanding system dynamics and the strength of chaos by plotting the values of the largest Lyapunov exponent $\lambda_1$ as a function of initial conditions of state variables. Therefore, a qualitative and quantitative description of chaos is possible.  For our purpose, it is sufficient to plot the largest Lyapunov exponent to distinguish the regions with chaotic motion from the regular ones. 
In Fig.~\ref{fig:lyap_polar1}, we present the polar plots of Lyapunov exponents diagrams for fixed values of the parameters
\begin{equation}
\label{eq:parki_2}
\mu=3,\quad \delta=1,\quad  \text{with}\quad \omega\in\{1,1.6,6,100\},
\end{equation}
and initial conditions were chosen as
\begin{equation}
\label{eq:ini_2}
\begin{split}
\ell(0)\in[0.001,5],\quad \varphi_1(0)=0,\quad \varphi_2(0)\in[-\pi,\pi],\quad v(0)=\omega_1(0)=\omega_2(0)=0.
\end{split}
\end{equation}
Here, $\ell(0)$ and $\varphi_2(0)$ are treated as the control parameters. The colorful diagrams visible in Fig.~\ref{fig:lyap_polar1} were obtained by numerically computing Lyapunov’s exponents on a grid of $400\times 400$ values of $(\ell(0),\varphi_2(0))$ and then plotted in the polar-plane. Therefore, in the radial direction, we measure the ratio of the initial lengths of the pendulums, that is, $\ell(0)$, while in the angular direction values of $\varphi_2(0)$ are given.  The color scale is proportional to the magnitudes of the largest Lyapunov exponent $\lambda_1$.

 The diagrams presented in Fig.~\ref{fig:lyap_polar1} are very useful because they give information on how the change in the initial length of the spring $\ell(0)$ and the initial swing angle $\varphi_2(0)$
 affect the dynamics of the whole system. Moreover, we can easily estimate the range of initial conditions for which the motion of the system is regular. For example, for $\omega=1$, which indicates $gl_1=k/m_2$, the amount of area responsible for the chaotic motion of the system is prominent. For the initial values $\ell(0)\in(0,1)$, so the spring is highly compressed, the system performs complex dynamics for arbitrary values of $\varphi_2(0)\neq 0$. However, for $\ell(0)\in(1,3)$ and for sufficiently small amplitudes of $\varphi_2(0)\in(-0.4\pi,0.4\pi)$, the system performs regular, non-chaotic oscillations. 	On the other hand, for larger values of $\varphi_2(0)$, the chaotic motion occurs and has maximum strength around $\varphi_2(0)\approx \pm\pi$, as expected. For larger values of $\ell(0)$, so the spring is initially stretched, the motion of the system is highly complex even for very small initial amplitudes of $\varphi_2$. 
The rest of the Lyapunov diagrams visible in Fig.~\ref{fig:lyap_polar1}, show the relation between initial values of $(\ell(0),\varphi_2(0))$ and further increasing values of $\omega$.  We see that
for higher values of $\omega$, the system is less chaotic since $\lambda_1$ is decreasing.  Moreover, 
the growth of areas responsible for the regular oscillations is visible. In particular, for $\omega=100$, there are no signs of chaos since $\lambda_1\approx 0$, for every initial condition. The above suggests the integrability of the system for $\omega\to \infty$. This particular case will be examined in more detail in Section \ref{sec:poin}. 

  	\subsection{Bifurcation diagrams and Lyapunov exponents \label{sec:sec_bif}}
  \begin{figure*}[t]	\centering
	\subfigure[$\ell(0)=3$]{
	\includegraphics[width=0.35\linewidth]{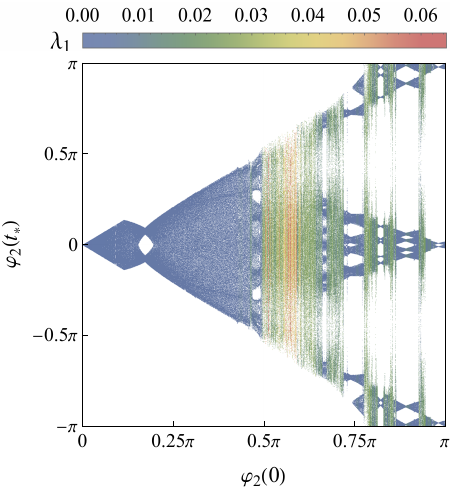}}	\hspace{1.5cm}\subfigure[$\vartheta_{2}(0)=0.75\pi$]{
	\includegraphics[width=0.35\linewidth]{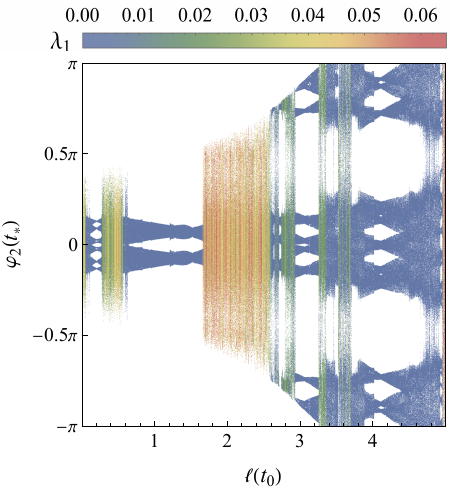}}
	\caption{(Color online)  Phase-parametric diagrams of system~\eqref{eq:vv} versus: a) the initial swing angle $\varphi_2(0)\in[0,\pi]$, b) the ratio of the initial lengths of the pendulums $\ell(0)\in [0.001,5]$. The initial conditions and values of the parameters are taken from Fig.~\ref{fig:lyap_polar1}(b). There are two cases to consider: a) we move in the angular direction of the Lyapunov diagram~\ref{fig:lyap_polar1}(b) for $\ell(0)=3$ and $\varphi_2(0)\in[0,\pi]$; b) we choose the initial swing angle to be $\varphi_2(0)=0.75\pi$ and move in the radial direction of the Lyapunov diagram~\ref{fig:lyap_polar1}(b) for $\ell(0)\in[0.001,5]$. Here $\varphi_2(t_\star)$ are the values of~$\varphi_2$, when the trajectory crosses the section plane $\ell=2.6$, for some~$t_\star$. The diagram is combined with the largest Lyapunov exponent $\lambda_1$ and the color scale is proportional to its magnitude. Very good agreement of the phase-parametric diagram with the Lyapunov diagram~\ref{fig:lyap_polar1}(b) is observed. 
 The coexistence of periodic, quasi-periodic, and chaotic orbits together with ,,periodic windows'' between chaotic layers is visible.   Exemplary periodic, quasi-periodic, and chaotic orbits are plotted in Fig.~\ref{fig:periodicki}}
	\label{fig:bifek}
\end{figure*}
  
\begin{figure*}[t]	\centering
	\subfigure[$\ell(0)=0.2$]{
	\includegraphics[width=0.32\linewidth]{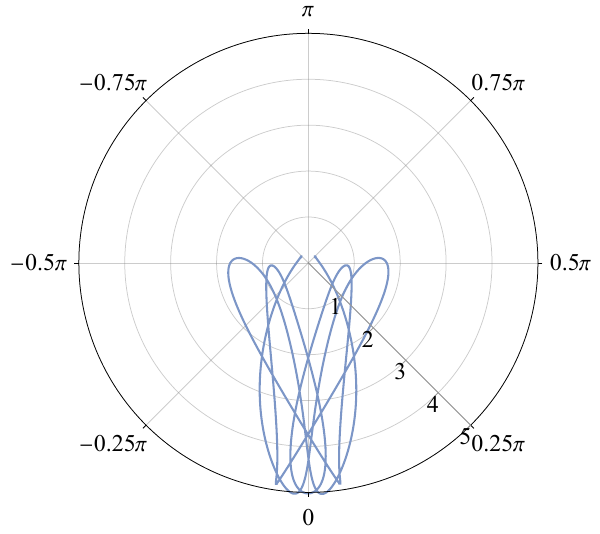}}
		\subfigure[$\ell(0)=4.5$]{
	\includegraphics[width=0.32\linewidth]{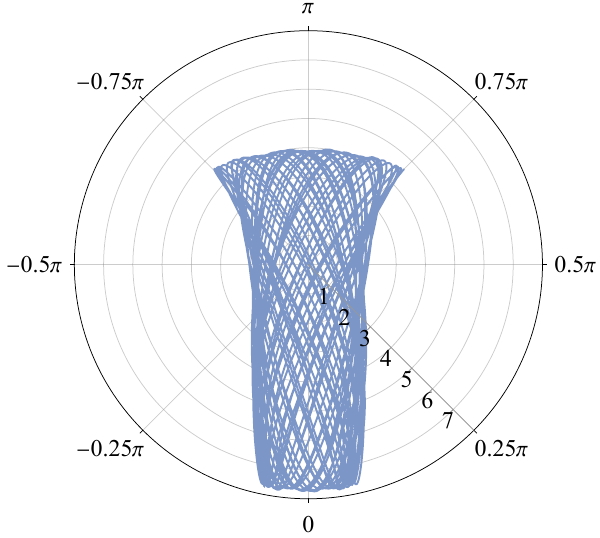}}
	\subfigure[$\ell(0)=2$]{
	\includegraphics[width=0.32\linewidth]{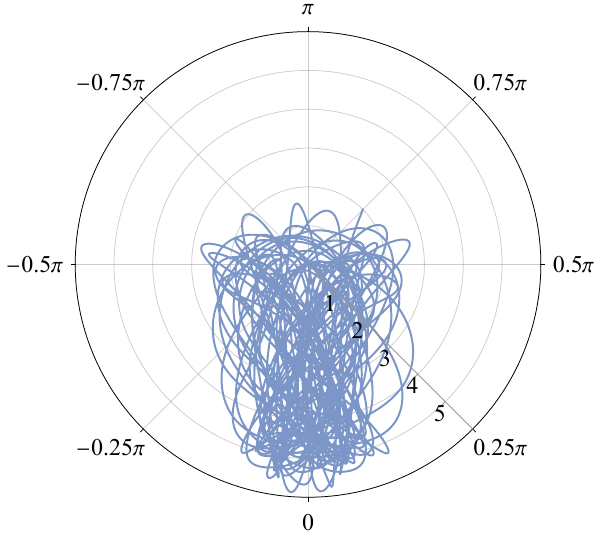}}
	\caption{(Color online)  Polar plots of exemplary a) periodic, b) quasi-periodic, and c) chaotic orbits of the double spring pendulum. The parameters values and initial conditions were taken from the phase-parametric diagram presented in Fig.~\ref{fig:bifek}(b). In the radial direction, we measure $\ell(t)$; in the angular direction, we measure $\varphi_2(t)$}
	\label{fig:periodicki}
\end{figure*}

\begin{figure}[http]
	\centering
	\includegraphics[width=0.9\linewidth]{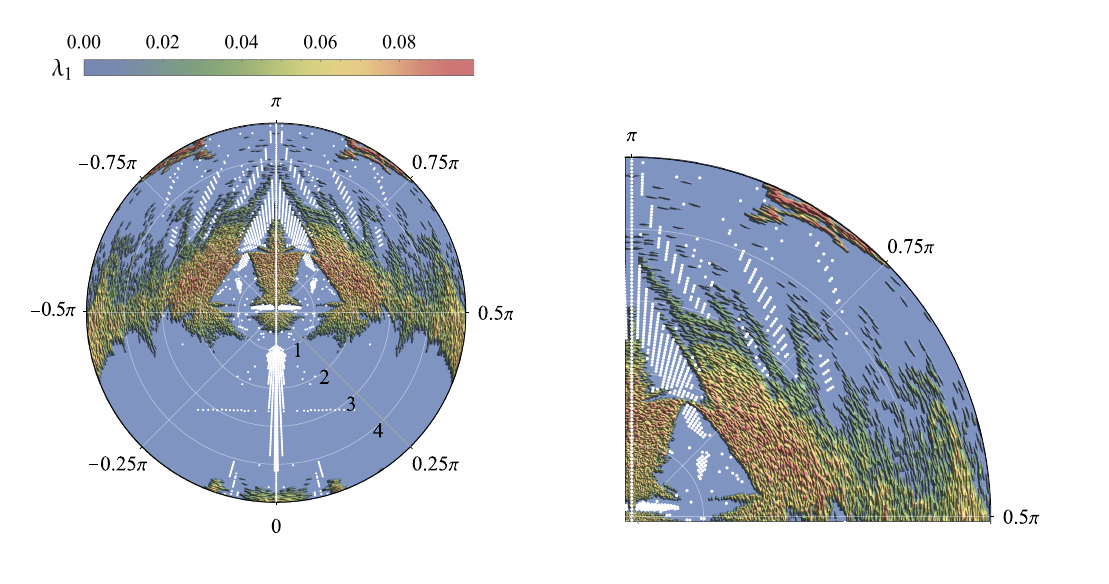}
	\caption{(Color online) The Lyapunov diagram visible in Fig.~\ref{fig:lyap_polar1}(b) with depicted periodic orbits. Each white point corresponds to distinct initial conditions	 $(\ell(0),\varphi_2(0))$  for which the motion of the spring pendulum is periodic with certain ratios of frequencies. The  polar plots of exemplary periodic orbits are depicted in Fig.~\ref{fig:periodicki_polar}	\label{fig:bif_pik}}
\end{figure}
The Lyapunov exponents is the essential tool for giving the quantitative description of chaos. However, in the context of the Hamiltonian system, this method has one weak point. Namely, in the blue regions of the Lyapunov diagrams as in Fig.~\ref{fig:lap_parki3}-\ref{fig:lyap_polar1}, where $\lambda_1\approx 0$, we are unable to distinguish values of the control parameters, for which the motion of the system is periodic or quasi-periodic.
This lack of distinction is a significant inconvenience since knowledge about the existence of resonance orbits in Hamiltonian dynamical systems is crucial. Nevertheless, constructing phase-parametric diagrams is an effective way to solve this problem.

 The phase-parametric diagram gives a
qualitative description of the system dynamics by plotting a state variable as a function of a suitably chosen control parameter~\cite{Szuminski:23::,Szuminski:20::}. However, in the context of Hamiltonian systems, it is convenient to show intersections of phase curves with a properly chosen surface of section. In this way, periodic, quasiperiodic, and chaotic orbits are distinguishable.  
 As the treated system is a spring-pendulum system, we use initial $\varphi_2(0)$ and $\ell(0)$ as bifurcation parameters. As the cross-section plane, we chose the length of the spring-mass at the equilibrium, that is, $\ell=\delta+\omega$.

\begin{figure*}[t]	\centering
	\subfigure[$\ell(0)=4.52, \varphi_2(0)=0.26$]{
	\includegraphics[width=0.32\linewidth]{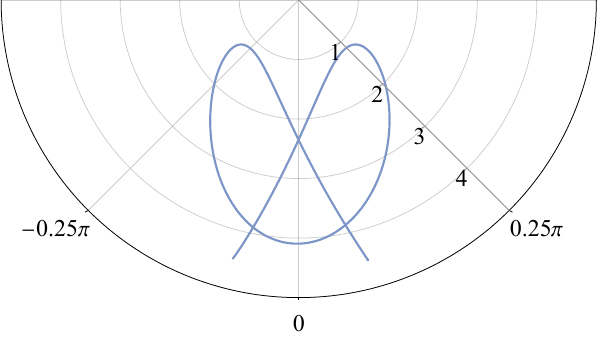}}
	\subfigure[$\ell(0)=1.01, \varphi_2(0)=1.52$]{
	\includegraphics[width=0.32\linewidth]{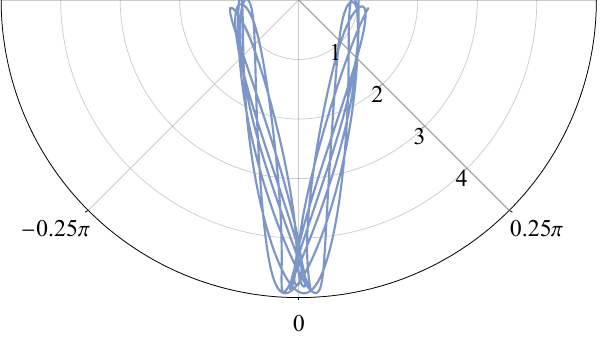}}
		\subfigure[$\ell(0)=3.37, \varphi_2(0)=0.7\pi$]{
	\includegraphics[width=0.32\linewidth]{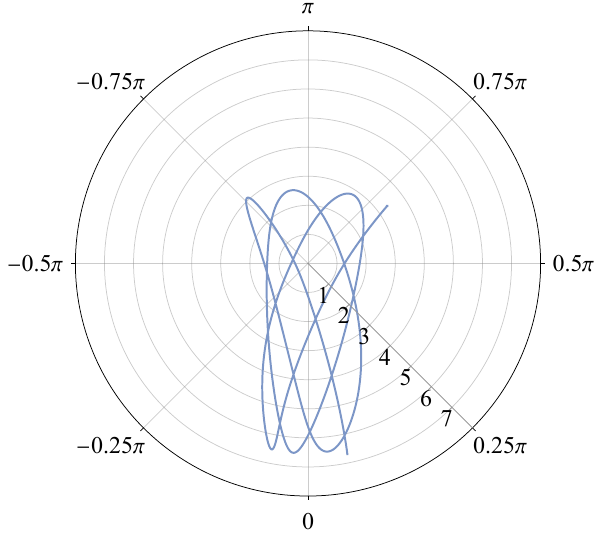}}
	\caption{(Color online)  Polar plots of exemplary periodic  orbits of the double spring pendulum. In the radial direction, we measure $\ell(t)$; in the angular direction, we measure $\varphi_2(t)$.  Each set of initial conditions $(\ell(0),\varphi_2(0))$ corresponds to a specific white dot marked on the Lyapunov diagram presented in Fig.~\ref{fig:bif_pik}}
	\label{fig:periodicki_polar}
\end{figure*}

\begin{figure}[http]
	\centering	
	\subfigure[$\omega=1$]{  \label{fig:per_1}
	\includegraphics[width=0.38\linewidth]{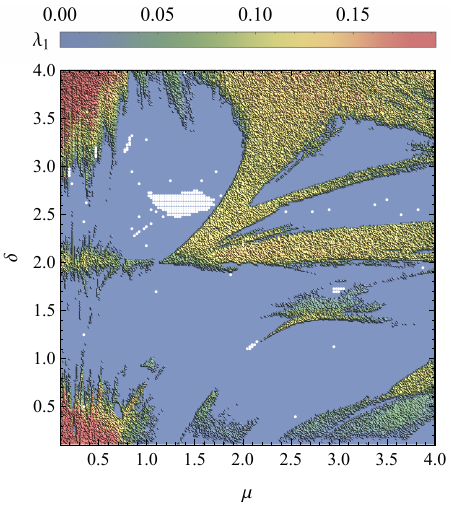}}
	 \hspace{1.5cm}
	 \subfigure[$\omega=0$]{ \label{fig:per_2}
	 	\includegraphics[width=0.38\linewidth]{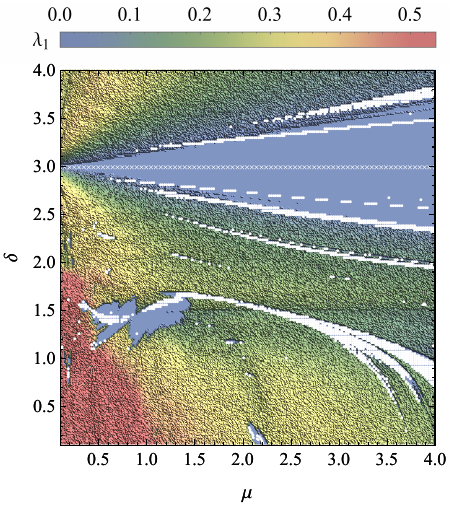} }
	\caption{(Color online) The Lyapunov diagrams presented in Figs.~\ref{fig:lap_parki3}-\ref{fig:lap_parki4} with depicted points corresponding to values of the parameters $(\delta,\mu)$ for which motion of the spring pendulum is periodic with certain ratios of frequencies. The polar plots of exemplary periodic orbits are shown in Figs.~\ref{fig:periodicki2}-\ref{fig:periodicki3} \label{fig:per}}
\end{figure}
\begin{figure*}[htp]	\centering
	\subfigure[$\mu=0.35,\ \delta=0.5$]{
	\includegraphics[width=0.32\linewidth]{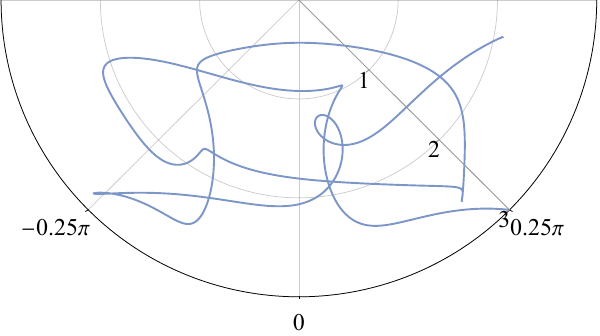}}
		\subfigure[$\mu=0.8,\ \delta=3.2$]{
	\includegraphics[width=0.32\linewidth]{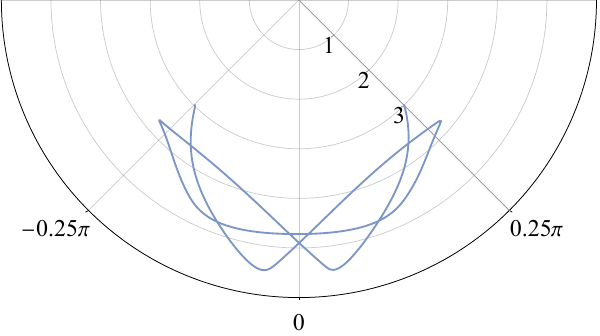}}
	\subfigure[$\mu=1.1,\ \delta=1.7$]{
	\includegraphics[width=0.32\linewidth]{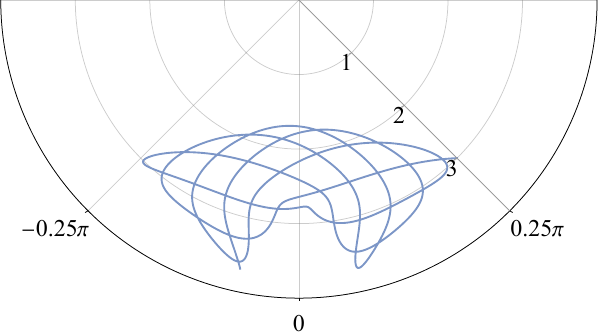}}\\
		\subfigure[$\mu=2.55,\ \delta=0.4$]{
	\includegraphics[width=0.32\linewidth]{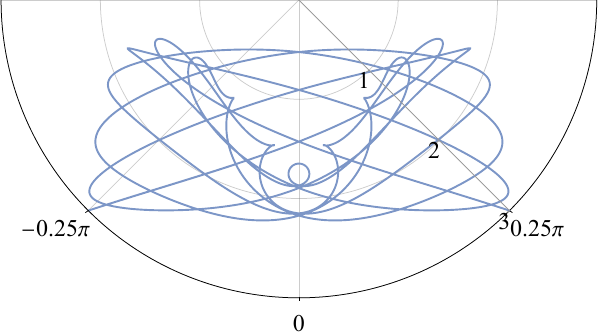}}
	\subfigure[$\mu=3,\ \delta=1.725$]{
	\includegraphics[width=0.32\linewidth]{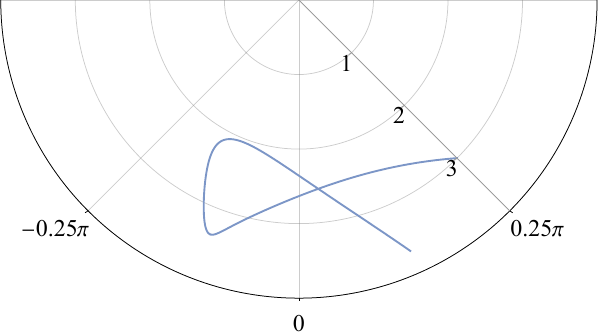}}
		\subfigure[$\mu=3,\ \delta=1.725$]{
	\includegraphics[width=0.32\linewidth]{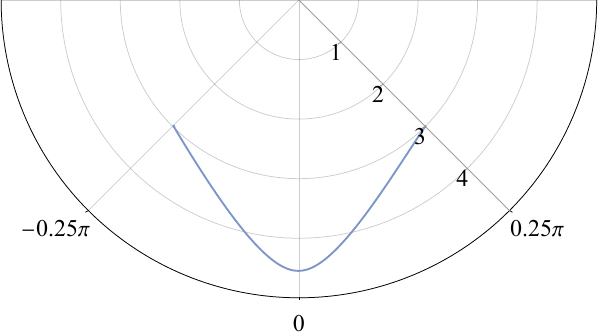}}
	\caption{(Color online)  Polar plots of exemplary periodic  orbits of the double spring pendulum. The parameters values $(\mu,\delta)$ and initial conditions were taken from the Lyapunov diagram presented in Fig.~\ref{fig:per_1}. In the radial direction, we measure $\ell(t)$; in the angular direction, we measure $\varphi_2(t)$}
	\label{fig:periodicki2}
\end{figure*}

Fig.~\ref{fig:bifek} presents the phase-parametric diagrams of the system computed for two one-parameter families of initial conditions $(\ell(0),\varphi_2(0))$ taken from Fig.~\ref{fig:lyap_polar1}(b). In Fig.~\ref{fig:bifek}(a), we move in the angular direction of the Lyapunov diagram~\ref{fig:lyap_polar1}(b) with $\ell(0)=3$ and $\varphi_2\in[0,\pi]$, while Fig.~\ref{fig:bifek}(b) corresponds to the case that we chose $\varphi_2(0)=0.75\pi$ and we move in the radial direction of the Lyapunov diagram~\ref{fig:lyap_polar1}(b) with $\ell(0)\in[0.001,5]$.  That is, for a given initial condition, we numerically integrate the equations of motion~\eqref{eq:vv} and build the diagram by collecting points $\varphi_2(t_\star)$ when $\ell(t)=2.6$, with varying values of initial $\varphi_2(0)$ and $\ell(0)$. Furthermore, to make the analysis more effective, we joined the phase-parametric diagrams with Lyapunov exponents computed previously in Fig.~\ref{fig:lyap_polar1}(b). The color scale is proportional to the magnitude of~$\lambda_1$.
At first,  we observe an excellent correspondence of the phase-parametric diagram with $\lambda_1$. The blue regions with $\lambda_1\approx 0$ correspond to regular oscillations, while other areas are responsible for the chaotic behavior of the system.  Moreover, while not visible via Lyapunov exponents diagrams, from the regular regimes, we can distinguish periodic orbits from quasi-periodic ones. 
  Fig.~\ref{fig:bifek}(a) starts with a regular regime corresponding to periodic and quasi-periodic oscillations of the spring near the equilibrium.  Further increasing values of the initial swing angle $\varphi_2$ indicate that the regular pattern divergences and in the neighborhood of the point $\varphi_2(0)=\pi/2$ the motion becomes chaotic with non-zero values of $\lambda_1$. However, even for very high values of $\varphi_2(0)$, we can still find periodic solutions located in the small gaps between completely chaotic regions.   
Fig.~\ref{fig:bifek}(b) presents a quite different picture.
The phase parametric diagram indicates the chaotic motion of the system for a very small initial length of the spring $\ell(0)\in[0.001,0.1]$. This was quite expected because the natural length of the spring is $\delta=1$. Therefore, for the initial length
$\ell(0)\in[0.001,0.1]$, the spring is highly compressed. Surprisingly moving further to the right, in the neighborhood of the point $\ell(0)=0.2$, the periodic solution appears. However, we claim that this periodic solution is unstable because, in the right neighborhood of this, we observe the rise of chaotic behavior of the system over the regime $\ell(0)\in[0.3,0.63]$.  For higher values of the control parameter $\ell(0)$, the motion is regular up to the point $\ell(0)\approx 1.7$, where the hyperchaos takes place in the wide range of $\ell(0)\in [1.7, 2.65]$. In addition, regular and chaotic regions coexist; As the initial length $\ell(0)$ increases, the appearance of periodic windows between chaotic layers is visible. 
For better understanding, in Fig.~\ref{fig:periodicki} polar plots of exemplary nonsingular orbits of the system with different initial conditions are taken from the phase-parametric diagram illustrated in Fig.~\ref{fig:bifek}(b).

As we have already seen, the Lyapunov exponents is an essential tool for giving a qualitative description of chaos. We used this method to specify areas on the two-parameter diagrams for which motion of the system is regular or chaotic. 
On the other hand, computations of phase-parametric diagram are effective in finding periodic orbits and their number or presenting the routes to the chaos of a given dynamical system. Therefore, it is reasonable to combine these two methods more systematically. 
Let us show this in the example of parameters corresponding to the Lyapunov diagram~\ref{fig:lyap_polar1}(b). First, we select from the diagram these pairs of $(\ell(0),\varphi_2(0))$, such that $\lambda_1=0$. Hence,  we obtain a grid of $n$ initial conditions for which the motion of the system is nonchaotic. Next,  for each element of the gird, we compute the phase parametric diagram by collecting points $\varphi_2(t_\star)$
 periodically, when $\ell(t_\star)=\delta+\omega=2.6$ for a wide range of $t$. As a result, we obtain $n$ lists of points $\varphi_2(t_\star)$.  Then,  in each list, we look for a scheme of repeating values of $\varphi_2(t_\star)$ in a certain order. The above can be effectively handled with the help of Mathematica. In this way, the rough distinction between periodic solutions from quasi-periodic ones is possible.  The result is presented in Fig.~\ref{fig:bif_pik}, where the periodic points are depicted in the Lyapunov diagram~\ref{fig:lyap_polar1}(b). Now, we have the complete picture of the system dynamics by specifying intervals where the motion is chaotic, quasi-periodic, or periodic.  In fact, each white point corresponds to different initial conditions $(\ell(0),\varphi_2(0))$ for which the spring pendulum motion is periodic with certain frequency ratios. 
 In Fig.~\ref{fig:periodicki_polar}, we show the polar graphs showcasing exemplary periodic orbits of the system. 
 
In Fig.~\ref{fig:per}, we present the Lyapunov diagrams previously shown in Figs.~\ref{fig:lap_parki3}-\ref{fig:lap_parki4}, depicting periodic points. Specifically, these diagrams illustrate the values of the parameters $(\delta, \mu)$ for which the system, under the initial condition~\eqref{eq:ini}, moves periodically with certain frequency ratios. In Fig.~\ref{fig:per_1}, a large cluster of periodic solutions can be observed in the left center of the diagram, along with several scattered periodic solutions throughout. The polar plots of periodic orbits for given values of the parameters $(\mu, \delta)$ are depicted in Fig.~\ref{fig:periodicki2}. Fig.~\ref{fig:per_2} exhibits a different pattern, where instead of one cluster of periodic solutions, we observe a myriad of periodic solutions forming resonance-like curves. The crosses in the diagram at $\delta=3$ represent the equilibrium points of the system. The abundance of periodic solutions can be attributed to the fact that, for $\omega=0$, the system is a two-degree-of-freedom model, which simplifies its complexity. Fig.~\ref{fig:periodicki3} presents polar plots of exemplary periodic orbits for prescribed values of the parameters $(\mu, \delta)$.

	\subsection{Poincar\'e cross-sections and Lyapunov exponents \label{sec:poin}}
	As evidenced in Fig.~\ref{fig:lap_parki4}, for $\omega=0$, one additional Lyapunov exponent is zero for arbitrary values of the remaining parameters $(\delta, \mu)$. This occurs because, for $\omega=0$, there exists an additional first integral
	\begin{equation}
		F=\partial_{\omega_1} L+\partial_{\omega_2} L=\const.
	\end{equation}
	Hence, the total angular momentum of the pendulums with respect
to the suspension point is the conserved quantity. Therefore, it is
reasonable to introduce new variables defined by
	\begin{equation}
	\begin{split}
		&\varphi(t)=\varphi_1(t),\quad \vartheta(t)=\varphi_1(t)-\varphi_2(t),\\
		&\omega(t)=\omega_{1}(t),\quad \Omega(t)=\omega_{1}(t)-\omega_{2}(t).
		\end{split}
	\end{equation}
At the level $F=c$, the system reduces to a system of two degrees of freedom with the following Lagrange function	\begin{equation}
		\begin{split}
		\label{eq:lag_red}
			L&=T-V_\text{eff},\\
			T&=\frac{1}{2}\left(1-\frac{\sin^2\vartheta}{1+\mu+2\ell\cos\vartheta+\ell^2}\right) v^2+\frac{1}{2}\left(\frac{\mu+\sin^2\vartheta}{1+\mu+2\ell\cos\vartheta+\ell^2}\right)\ell^2\Omega^2 - \left(\frac{(\cos\vartheta+\ell)\ell \sin\vartheta    }{1+\mu+2\ell\cos\vartheta+\ell^2}\right)v\Omega,\\ 
			V_\text{eff}&=\frac{c^2}{2(1+\mu+2\cos\vartheta+\ell^2)}+\frac{1}{2}\left(\delta-\ell\right)^2.
		\end{split}
	\end{equation}

	As the motion of the reduced system takes place in four-dimensional phase space, we can effectively compute the Poincar\'e cross-sections. 	The main idea of the Poincar\'e cross sections is very simple. We consider a surface (in our case $\ell_0=1$) in the phase space that is traversed by all trajectories. Together with the energy conservation $E= T+V_\text{eff}$,   we have a two-dimensional surface embedded in the  $(\vartheta,\Omega,v)$-space. This surface is formed by two connected components~$v_{\pm}=v_{\pm}(E,\vartheta,\Omega)$, which are two roots of the quadratic equation for $v$ imposed by energy conservation $E=\const$. For convenience, we choose~$v>0$, and as the coordinates on this surface, we take~$(\vartheta,\Omega)$.   As a result, we obtain a pattern in the plane, which can be easily visualized and interpreted. If the motion of the system is periodic, the trajectory passes through the plane only in a finite number of intersections. The quasi-periodic motion is manifested by continuous loops on the section plane, while a chaotic trajectory intersects the section plane in scattered, random-looking points. 
	\begin{figure*}[t]	\centering
	\subfigure[$\mu=0.3,\ \delta=1.26$]{
	\includegraphics[width=0.32\linewidth]{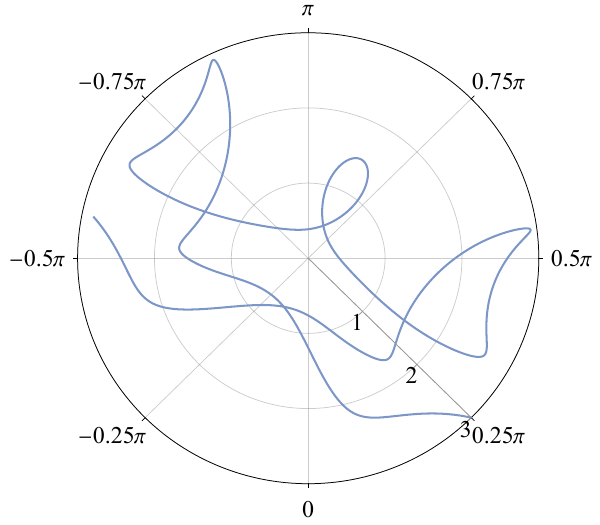}}
		\subfigure[$\mu=0.1,\ \delta=1.3$]{
	\includegraphics[width=0.32\linewidth]{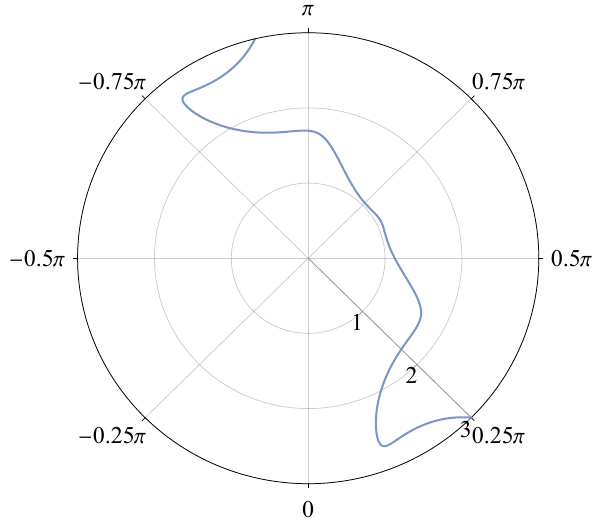}}
			\subfigure[$\mu=2.1,\ \delta=0.82$]{
	\includegraphics[width=0.32\linewidth]{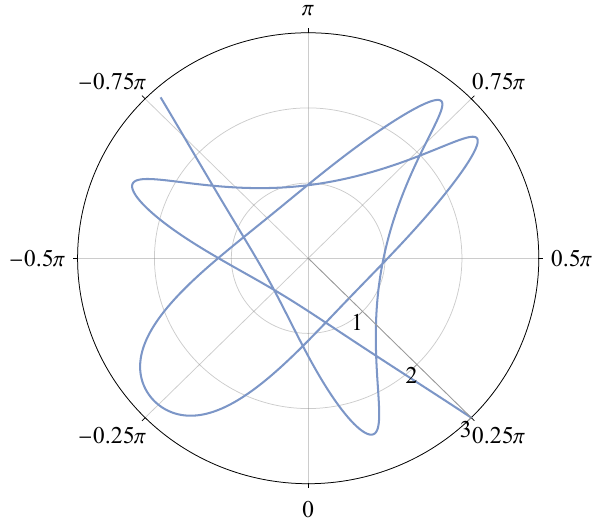}}\\
		\subfigure[$\mu=1.78,\ \delta=1.62$]{
	\includegraphics[width=0.32\linewidth]{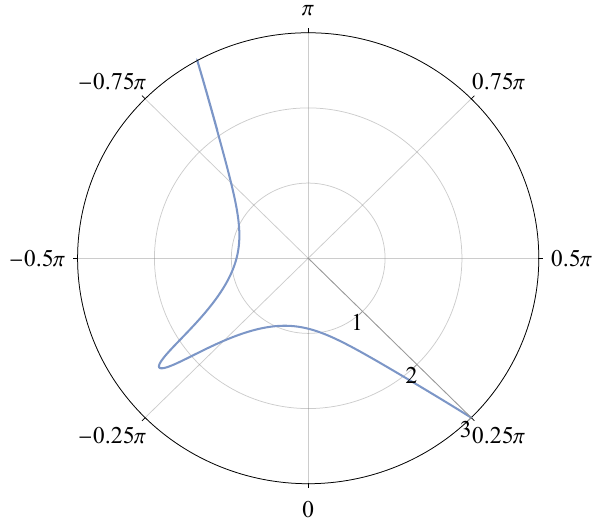}}
		\subfigure[$\mu=1.54,\ \delta=2.02$]{
	\includegraphics[width=0.32\linewidth]{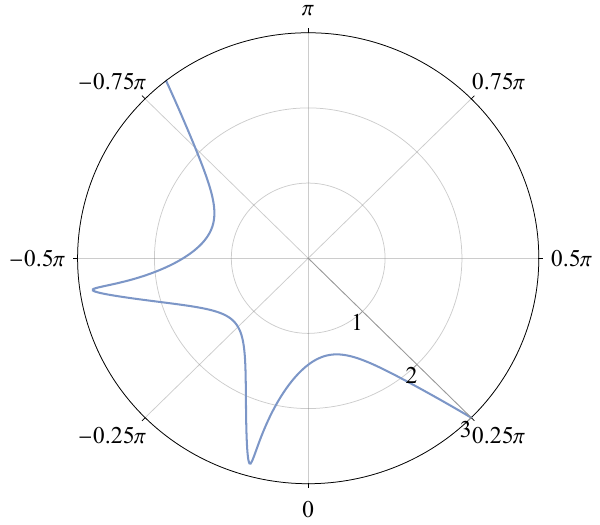}}
			\subfigure[$\mu=2.9,\ \delta=1.86$]{
	\includegraphics[width=0.32\linewidth]{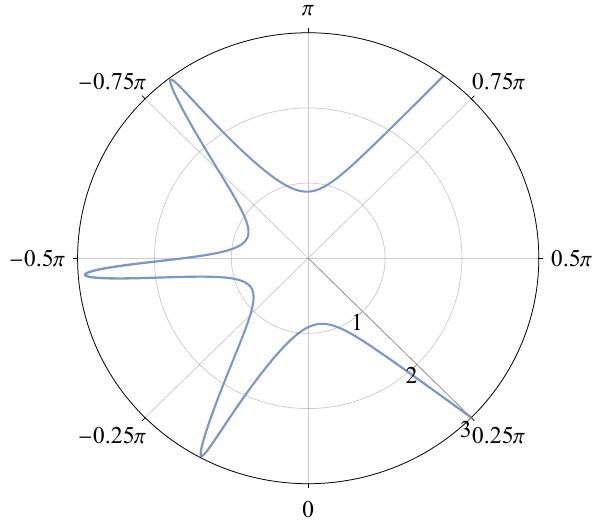}}
	\caption{(Color online)  Polar plots of exemplary  periodic  orbits of the double spring pendulum. The parameters values $(\mu,\delta)$ and initial conditions were taken from the Lyapunov diagram presented in Fig.~\ref{fig:per_2}. In the radial direction, we measure $\ell(t)$; in the angular direction, we measure $\varphi_2(t)$}
	\label{fig:periodicki3}
\end{figure*}
		\begin{figure*}[htp]
		\centering
		\subfigure[Energy $E=E_0+0.001$, Regular domain with the chaotic remains of a separatrix]{
		\includegraphics[width=0.85\linewidth]{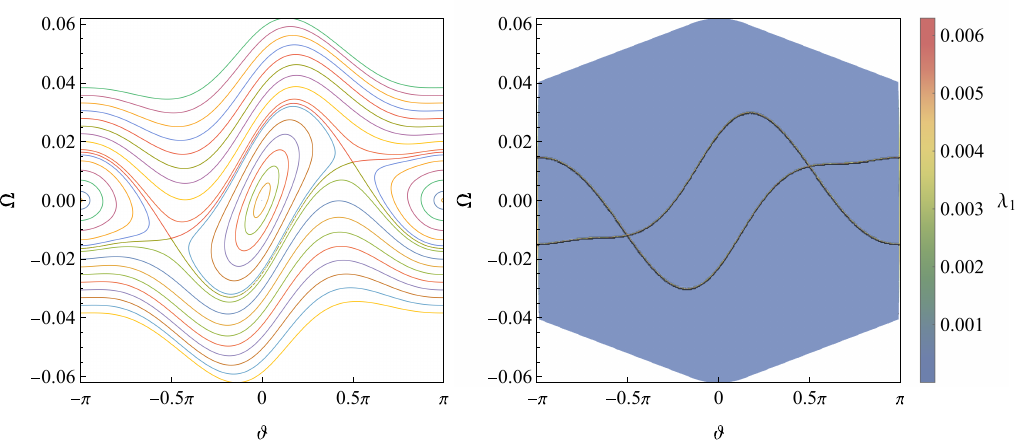}}\\
			\subfigure[Energy $E=E_0+0.017$, the beauty of the coexistence of periodic, quasi-periodic, and chaotic orbits]{
		\includegraphics[width=0.85\linewidth]{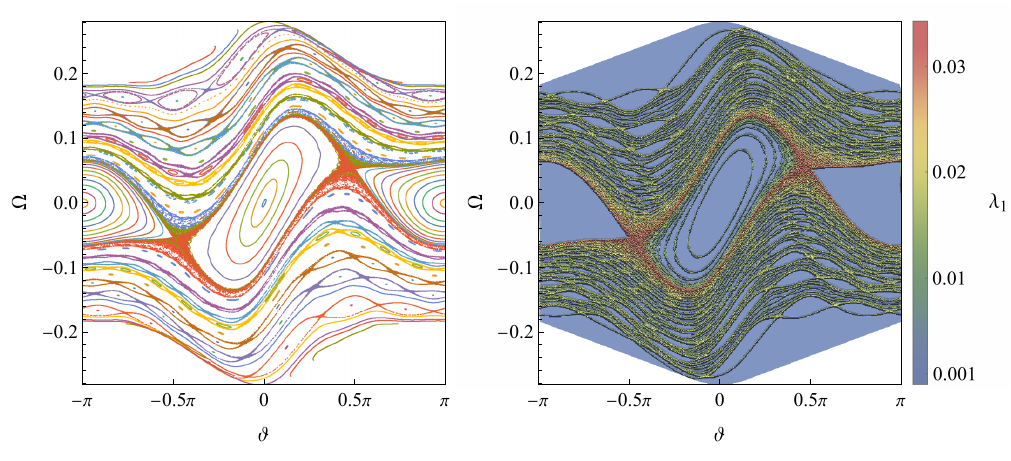}}\\	
			\subfigure[Energy $E=E_0+0.05$, the highly chaotic stage system's dynamics]{\includegraphics[width=0.85\linewidth]{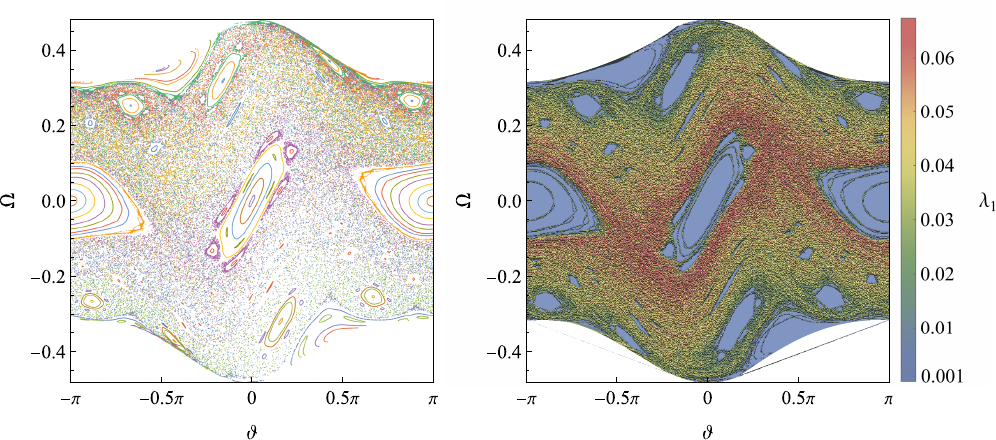}}\\ 		\caption{(Color online)  The Poincar\'e sections of the system~\eqref{eq:lag_red} and their corresponding Lyapunov diagrams, constructed for constant values of the parameters $\mu=3,\ \delta=1,$ and $c=0$ with a gradually increasing energy level $E$, where $E_0$ is the energy minimum. The cross-section plane was defined as $\ell=1$ with the direction $v>0$. 
In the Poincar\'e sections, the colors are associated with different classical trajectories. For the Lyapunov exponents, the color code is given on the bar.  Blue regions indicate regular dynamics, while regions with $\lambda_1>0$ correspond to the system's chaotic behavior}
		\label{fig:ciecie1}
	\end{figure*}
	
	\begin{figure}[htp]
		\centering
		\includegraphics[width=0.48\linewidth]{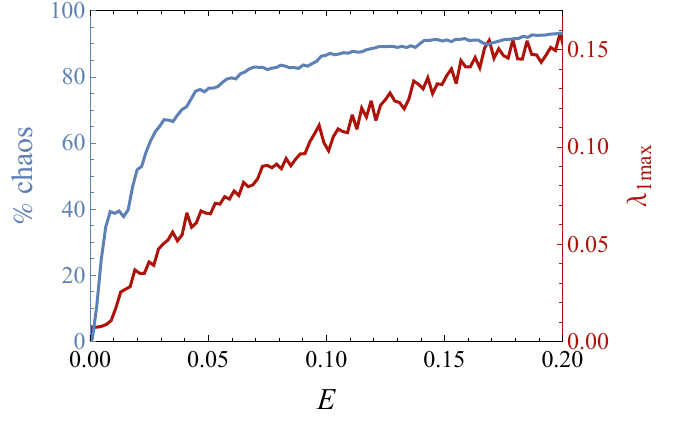}\hspace{0.5cm}
		\includegraphics[width=0.48\linewidth]{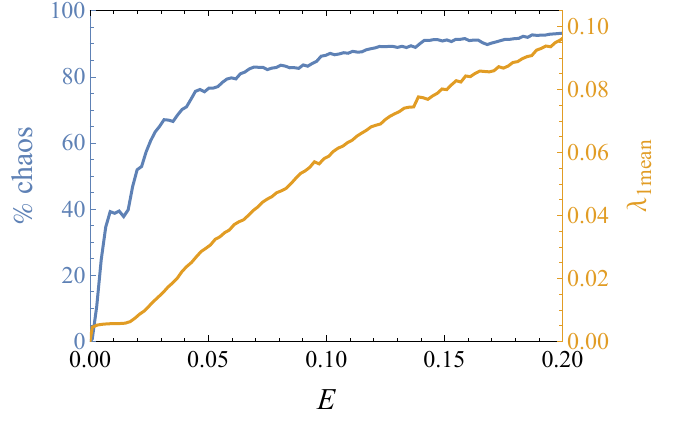}\\ 		\caption{(Color online)  The percentage of chaos against the maximal  and mean  values of the largest Lyapunov exponent $\lambda_1$ in phase space as a function of the energy }
		\label{fig:mean_max}
	\end{figure}

Fig.~\ref{fig:ciecie1} shows six Poincar\'e sections and their corresponding two-parameter Lyapunov diagrams of system~\eqref{eq:lag_red}. These were constructed with gradually increasing values of the energy $E$, for the following constant parameters
	\begin{equation}
	\label{eq:pp}
	\mu=3,\quad  \delta=1,\quad c=0.
	\end{equation}
 It is easy to verify that the minimum energy of the system~\eqref{eq:lag_red}, corresponding to the equilibrium position of the pendulums when $\ell_0=1$ is equal to $E_0=0$.  Therefore, we may expect regular behaviors of the system for energies close to $E_0$. The Poincar\'e section visible in Fig.~\ref{fig:ciecie1}(a) confirms our suspicion. At the energy level $E=0.001$ it presents a very regular image. 
 We observe three particular periodic solutions located at $\vartheta=0,\pm\pi$ with $\Omega=0$,  surrounded
by invariant tori. Each closed loop corresponds to a quasi-periodic motion. These three particular solutions are visible in the remaining Poincar\'e sections depicted in Fig.~\ref{fig:ciecie1}(b-c) as well. This is because 
the equations of motion governed by Lagrange function~\eqref{eq:lag_red} have an invariant manifold
$
\scM=\left\{(\ell,v_\ell, \vartheta, \Omega)\in \field{R}^4 \ \big{|}\ \vartheta=\Omega=0\right\}.
$
The Lagrange function~\eqref{eq:lag_red} constrained to $\scM$, reduces to a system of one degree of freedom that defines the harmonic oscillator with solutions $\ell(t)=\ell_0+\sqrt{2E}\cos t,$ and $  v=\dot \ell(t)$ with $\vartheta(t)=\Omega(t)=0$. 
We may conclude that for such a small energy value the motion of the system is almost integrable. However, 
what is not visible in the Poincar\'e plane, 	
the corresponding Lyapunov diagram indicates that the curve, which separates librational motion from the rotational one, is indeed, weakly'' chaotic. So, we claim that this curve is a source of the chaotic motion of the system since $\lambda_1>0$.

	The situation becomes more complex when we increase the energy to the value $E=0.017$. Fig.~\ref{fig:ciecie1}(b) confirms the chaotic nature of the reduced system. Most of the invariant tori, responsible for the rotational motion of the system, decay giving rise to periodic orbits bounded by chaotic remains of the separatrices.  Indeed, the corresponding Lyapunov diagram shows a variety of decayed separatrices. This is one of the main benefits of combining Poincaré cross sections with Lyapunov diagrams. For large values of initial conditions  (in our cases $500\times 500$)
uniformly distributed in the available area of the Poincar\'e plane, with the help of Lyapunov exponents, we can detect  ,,chaotic folds" responsible for weak chaotic motion.  	  Thus, the Lyapunov exponents diagram serves as a complementary tool to the Poincar\'e  sections, providing insights into chaotic dynamics that may not be apparent from the latter, which can be constructed for a much smaller grid of initial conditions from visual reasons.
Fig.~\ref{fig:ciecie1}(c) shows the Poincar\'e section and the Lyapunov diagram constructed for $E=0.05$. As the energy value is increased, the area at the Poincar\'e plane responsible for the chaotic motion of the system also increases, as expected.  These random-looking points correspond to the fact that trajectories can freely wander over large regions of the phase space.  The trajectories are no longer confined to the surfaces of nested tori in contrast to the integrable system. The tori are destroyed and the trajectories begin to move outside them.
 The corresponding  Lyapunov diagram finally confirms that for  the prescribed values of the parameters  the reduced system governed by Lagrange function~\eqref{eq:lag_red} is 
highly not integrable Hamiltonian system. 
	
With the help of Lyapunov exponents, we can easily estimate the percentage of chaos equipped in the available area of the Poincar\'e plane as a function of the energy.  Indeed,  for a  large number of initial conditions uniformly distributed in the available area of the Poincar\'e plane determined by the energy value, we compute repeatedly the largest Lyapunov exponentss Next, we calculate the ratio of the number of points with a Lyapunov exponent different from zero  (in practice larger than  0.003)  to the total amount of initial conditions.   Likewise, we can calculate the maximum and mean values of the Lyapunov exponent for a given energy.
 Results for $E\in[0,0.2]$ are shown in Fig.~\ref{fig:mean_max}. 
As expected, for values of the energy close to the energy minimum $E_0=0$, the percentage of chaos is almost zero. This is in accordance with the Poincar\'e section visible in Fig.~\ref{fig:ciecie1}(a), where the weak chaos appears only in the area where the separatrix was located.
Next,  we have the classical transition from a regular regime at low energies to a highly chaotic regime at larger ones. Indeed,  the percentage of chaos increases rapidly, and at $E=0.05$ it covers $75\%$ of the available area of the Poincar\'e section plane visible in Fig.\ref{fig:ciecie1}(c). Then, for larger values of the energy, the system monotonically goes to an almost fully ergodic one with the percentage of chaos very close to $95\%$. The maximal and average values of $\lambda_1$ show a similar behavior starting from zero value at the energy minimum and then increasing linearly up to their highest values reached at $E=0.2$. 
	\begin{figure}[t]
	\subfigure[The standard Poincar\'e cross-section]{
		\centering\includegraphics[width=0.48\linewidth]{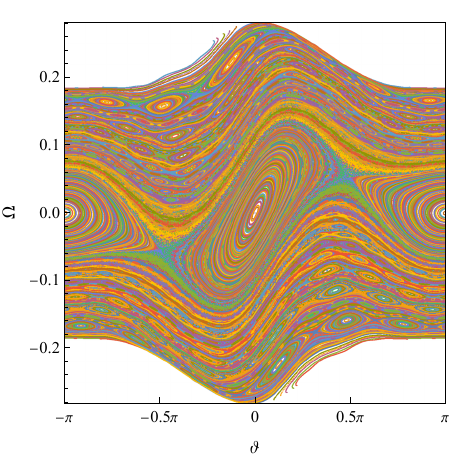}}\hspace{0.5cm}
	\subfigure[The filtered Poincar\'e cross-section]{
		\includegraphics[width=0.48\linewidth]{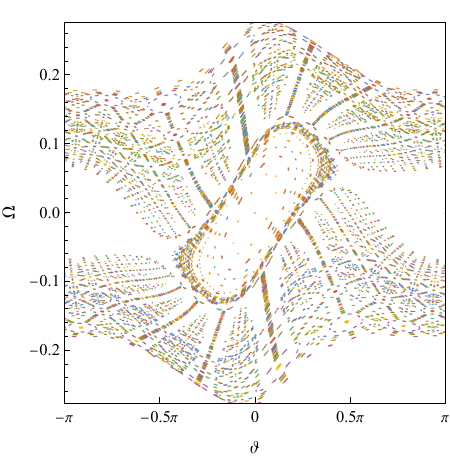}}
	\caption{(Color online) The  Poincar\'e sections of the system~\eqref{eq:lag_red}  constructed for:  a) $3\cdot 10^3$ and b)  $3\cdot 10^4$ initial conditions uniformly distributed at the available area. 
		The parameters were chosen by $\mu=3,\delta =1, c=0$, and
		the cross-section plane was defined as $\ell=1$ with the direction $v>0$. The left picture is unclear since periodic, quasi-periodic, and chaotic orbits overlap. The right Poincar\'e section is filtered from chaotic and quasi-periodic orbits, presenting the skeleton structure of stable periodic solutions only}
	\label{fig:ciecie_periodic}
\end{figure}		
As we have seen, the Poincar\'e cross-section method is an effective tool for giving a quick insight into the dynamics of the considered model. Indeed. Fig.~\ref{fig:ciecie1}(b) shows the beautiful co-existence of periodic, quasi-periodic, and chaotic motion.
To show the detailed structure of this co-existence, we have to increase the number of initial conditions and perform successive magnifications of the cross-section. 
However,  for a large number of initial conditions,  in the global cross-section, certain characteristic patterns are not visible. In Fig.~\ref{fig:ciecie_periodic}(a), we present the Poincar\'e sections constructed for $3\cdot 10^3$ initial conditions uniformly distributed in the available area of the Poincar\'e plane.  As we can see,   it is hard to distinguish in the global view periodic orbits from quasi-periodic and chaotic ones because the orbits are located very close to each other.  
Since Poincar\'e it is known that periodic solutions of the considered system form the skeleton of the structure of the phase space. Of course, there are many approaches for finding and studying periodic solutions; for instance, see the newest one~\cite{Lazarotto:24::}.  Here, we propose to locate stable periodic solutions in the global cross-section and use them as a kind of filter. 

To find periodic orbits and their numbers in the available area of the Poincar\'e plane,  we can use the analogous approach as in Section~\ref{sec:sec_bif}. We make the following.
We fix the value of the energy first integral $E$. Then, for a set $B$ containing a large number of initial conditions 
 uniformly distributed in the available area of the Poincar\'e plane, we compute the largest Lyapunov exponent $\lambda_1$. Elements in~$B$ with  $\lambda_1>0$ are  disregarded. In this way, we obtain a subset of  $n$ initial conditions for which the motion of the system is non-chaotic. Next, for each initial condition, we compute the Poincar\'e cross sections with a sufficiently large number of sections restricted to the plane $(\vartheta,\Omega)$. As a result, we obtain $n$ lists with section points $(\vartheta,\Omega)$. Then, we perform a similarity check. Namely, 
 in each list, we look for the scheme of repeated values of $(\vartheta,\Omega)$ in a certain order. If the motion is periodic, then we have a finite number of repeating points. The quasi-periodic motion is manifested by the infinite number of distinct points.    In this way, the rough but effective distinction between periodic and quasi-periodic orbits is possible. 	The results of exemplary computations for $E=0.017$ are presented in Fig.~\ref{fig:ciecie_periodic}(b). It is the Poincar\'e section plane, similar to the one shown in Fig.~\ref{fig:ciecie1}(b), but constructed for $3\cdot10^4$ initial conditions, where chaotic and quasi-periodic trajectories are eliminated, leaving only solutions close to stable periodic orbits. 
 This plot, which we call the filtered Poincar\'e cross section, shows the rich structure of the system dynamics which is hidden in a usual cross-section. We get
a vast number of periodic orbits of various rations of frequencies, even the ones corresponding to high-order resonances.  
This leads to the fact that the invariant tori, which are located near these high-order resonance orbits, are prone to decay with further increasing energy. This is visible in Fig.~\ref{fig:ciecie1}(c), where most of the tori decay into global chaos.  
 				\begin{figure}[t]
\begin{center}

				\includegraphics[width=0.48\linewidth]{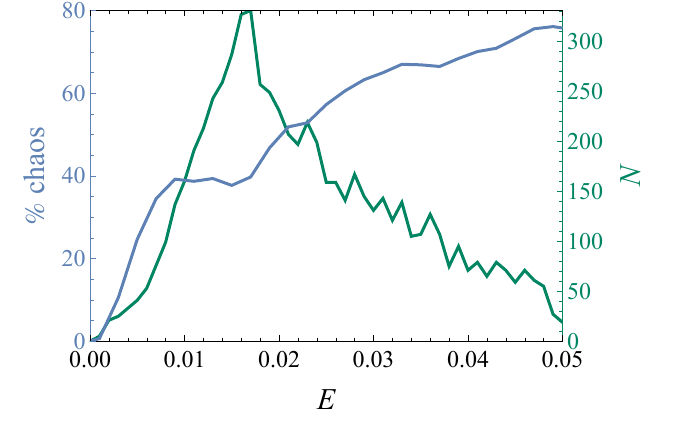}
				\caption{(Color online)  The percentage of chaos against the number of stable periodic orbits in the Poincar\'e section planes a function of energy }
		\label{fig:periodic}
		\end{center}
	\end{figure}
	\begin{figure*}[htp]
		\centering
						\subfigure[Mass ratio $\mu=5$, highly chaotic dynamics]{
		\includegraphics[width=0.85\linewidth]{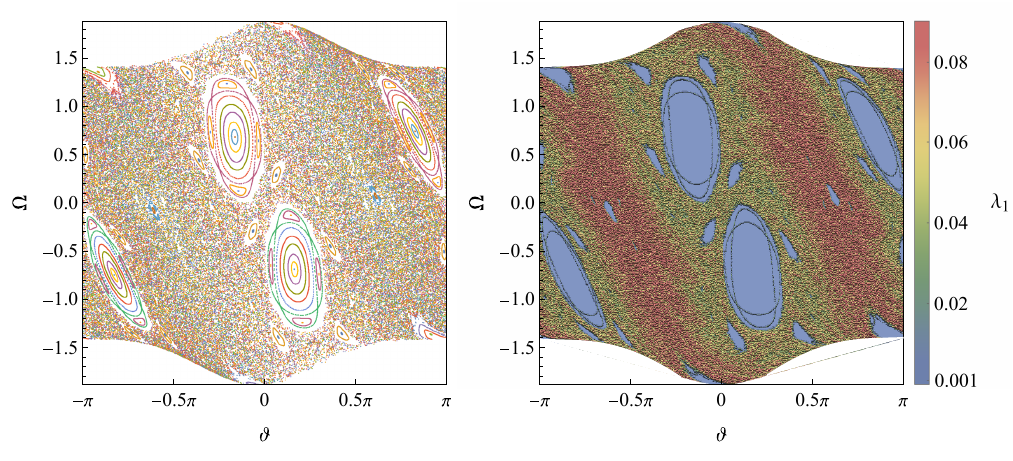}}
				\subfigure[Mass ratio $\mu=10$, the rise of regular islands between chaotic layers]{
		\includegraphics[width=0.85\linewidth]{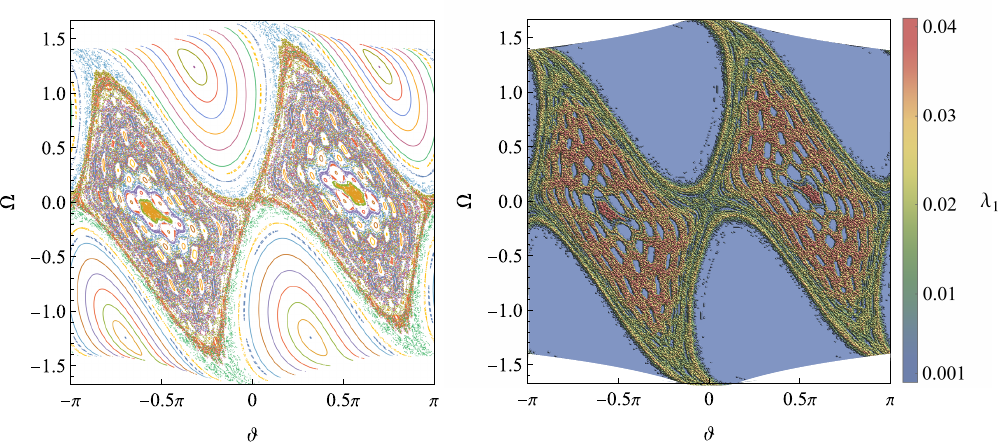}}\\
		\caption{(Color online) The Poincar\'e sections of the system~\eqref{eq:lag_red} and their corresponding Lyapunov diagrams, constructed for constant values of the parameters $\delta=1, \, c=0,$  at the energy level $E=E_0+0.5$  where $E_0=0$ is the energy minimum,  with gradually increasing values of the mass ratio $\mu$. The cross-section plane was defined as $\ell=1$ with the direction $v>0$. 
In the Poincar\'e sections, the colors are associated with different classical trajectories. For the Lyapunov exponents, the color code is given on the bar.  Blue regions indicate regular dynamics, while regions with $\lambda_1>0$ correspond to the system's chaotic behavior}
		\label{fig:ciecie2}
	\end{figure*}
		\begin{figure*}[htp]
		\centering				\subfigure[Mass ratio $\mu=15$, decreasing of chaotic area]{ 
		\includegraphics[width=0.85\linewidth]{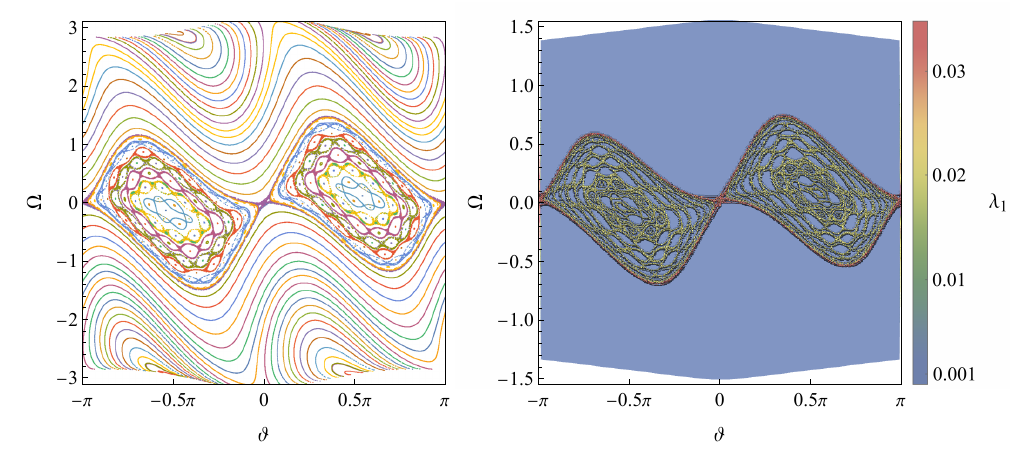}}\\
		\subfigure[Mass ratio $\mu=50$, the almost integrable dynamics
	]{	\includegraphics[width=0.85\linewidth]{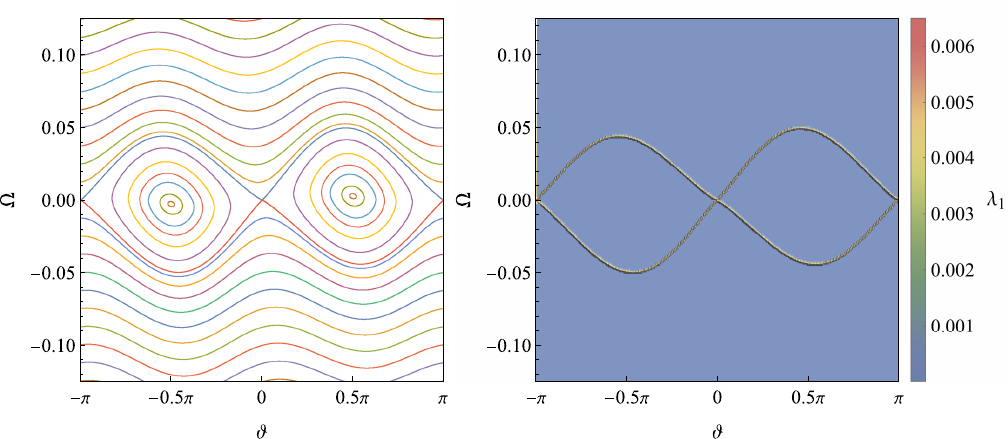}}
		\caption{(Color online)  The Poincar\'e sections of the system~\eqref{eq:lag_red} and their corresponding Lyapunov diagrams, constructed for constant values of the parameters $\delta=1, \, c=0,$  at the energy level $E=E_0+0.5$  where $E_0=0$ is the energy minimum,  with gradually increasing values of the mass ratio $\mu$. The cross-section plane was defined as $\ell=1$ with the direction $v>0$. 
In the Poincar\'e sections, the colors are associated with different classical trajectories. For the Lyapunov exponents, the color code is given on the bar.  Blue regions indicate regular dynamics, while regions with $\lambda_1>0$ correspond to the system's chaotic behavior}
		\label{fig:ciecie4}
	\end{figure*}

		\begin{figure}[htp]
		\centering
		\subfigure[The Poincar\'e section restricted to $(\varphi_2,\omega_2,\varphi_1)$-space]
		{
		\includegraphics[width=0.4\linewidth]{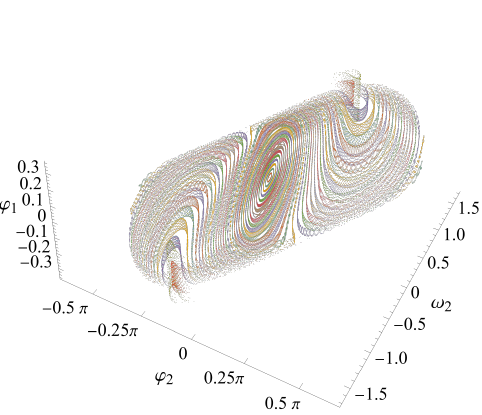}}\hspace{1.5cm}
		\subfigure[Projection to $(\varphi_2,\omega_2)$-plane]{
		\includegraphics[width=0.35\linewidth]{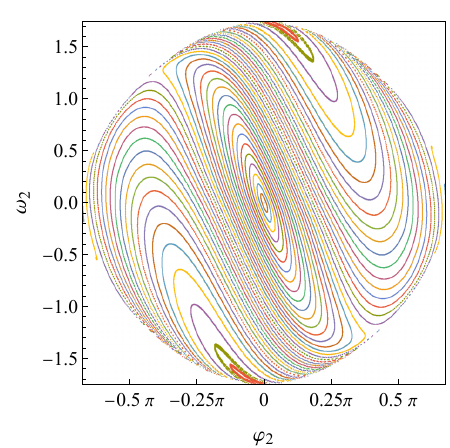}}\\ 		\caption{The Poincar\'e cross-section of the original system~\eqref{eq:vv} made for constant values of the parameters $\mu=3,\, \delta=1, \ \omega=100$.  The cross-section plane was defined as $\ell=1$ with the direction $v>0$. The plot indicates regular, almost integrable dynamics. Each color at the Poincar\'e plane corresponds to distinct initial conditions }
		\label{fig:ciecie5}
	\end{figure}
Repeating the above computations for a certain range of energy, we can estimate the number of stable periodic solutions as a function of the energy. In Fig.~\ref{fig:periodic}, we can see a representation of the relationship between the percentage of chaos and the number of stable periodic orbits available in the phase space. The left axis shows the percentage of chaos, while the right axis displays stable periodic orbits as a function of energy $E$. When the energy is close to the energy minimum, the percentage of chaos is almost zero, and we have several periodic orbits. However, for slightly larger values of the energy, the invariant tori broke up giving rise to periodic solutions. We observe that the number of periodic orbits and the percentage of chaos increase rapidly as energy increases as well. At the energy level $E=0.017$, the number of distinct periodic orbits is maximal with  $N=331$, while the chaos occupies  $40\%$ of the available   area of the Poincar\'e plane, see the corresponding Poincar\'e sections illustrated in Figs.~\ref{fig:ciecie1} and~\ref{fig:ciecie_periodic}.  
As the energy increases, the periodic orbits successively diverge into global chaos, making the available phase space almost fully ergodic.

Up to now, we have examined the behavior of the model by using Pincar\'e sections while altering the system's energy. However, it would be worthwhile to conduct a similar analysis keeping the energy constant and gradually increasing the mass ratio $\mu$.
In Figs.~\ref{fig:ciecie2}-\ref{fig:ciecie4}, we present four pairs of Poincar\'e sections and corresponding Lyapunov diagrams created for $\omega=\delta=1$ with varying values of $\mu$ at a constant energy level $E=E_0+0.5$, where $E_0=0$  represents the minimum energy at rest. As we can notice, for the relatively high value of the energy level, the values of $\vartheta$ span the whole range $\vartheta\in[-\pi, \pi]$ and the angular velocity $\Omega$ reaches high values in comparison to Fig.~\ref{fig:ciecie1}.  For the mass ratio $\mu=5$, the system reveals highly chaotic dynamics visible in terms of scattered points at the Poincar\'e section plane. Indeed, almost all invariant tori decay into the global chaos and the maximal Lyapunov exponent reaches $\lambda_1\approx 0.086$. The situation changes when we increase the value of $\mu$, for instance,  up to $\mu=10$. At the corners of the plane, illustrated in Fig.~\ref{fig:ciecie2}(b), we observe the emergence of four stable periodic solutions bounded by quasi-periodic loops. In the central region, a striking coexistence of periodic and chaotic chains becomes apparent. This is particularly evident in the associated Lyapunov diagram, where each chaotic fold is identified. 
As we increase the mass ratio up to $\mu=15$, the dynamics of the system become less chaotic which is clearly visible in Fig.~\ref{fig:ciecie4}(a). The region responsible for chaotic motion at the Poincar\'e plane also decreases with decreasing values of $\lambda_1$. This trend becomes even clearer when we further increase the mass ratio up to $\mu=50$. 
As we can notice, even for such a high value of energy, the dynamics of the system is very regular. In the  Poincar\'e section plane, illustrated in Fig.~\ref{fig:ciecie4}(b), we can observe the shapely elegant invariant tori suggesting the system's integrability. There are no signs of chaotic motion at all. Instead of that, we can detect two stable particular periodic solutions enclosed by the separatrix, which separates the librational and rotational motion.
However, looking at the corresponding Lyapunov diagram,  we notice that this remains of the separatrix is the source of chaos since $\lambda_1>0$, precluding the system's integrability.

After completing the numerical analysis, it is worth revisiting the Lyapunov exponents diagram shown in Fig.~\ref{fig:lyap_polar1}(d). For the parameters~\eqref{eq:parki} and $\omega=100$, the original system~\eqref{eq:vv} has zero Lyapunov exponents for any initial condition. In Fig.~\ref{fig:ciecie5}, we show the Poincar\'e cross-section of six-dimensional system~\eqref{eq:vv} restricted to  $(\varphi_2,\omega_2,\varphi_1)$-space and its projection to    $(\varphi_2,\omega_2)$-plane. As we can notice, for the sufficiently high value of $\omega$, which implies that the gravitational potential is much stronger than Hooke's potential, the double spring pendulum possesses integrable dynamics. In contrast to previously examined Poincar\'e sections, all six Lyapunov exponents are zero up to a value of $\lambda_1$ less than $0.002$. Additionally, due to the small amplitude of oscillations of the first pendulum, i.e., $\varphi_1\in[-0.06\pi, 0.06\pi]$, the projection of 3D  Poincar\'e section to $(\varphi_2,\omega_2)$-plane
 appears to be that 
of a like standard Poincar\'e section for a Hamiltonian system with two degrees of freedom. 
 			\section{Integrability analysis \label{sec:integrability}} The numerical analysis exemplified by the bifurcation diagrams, the Poincar\'e sections as well as by the Lyapunov exponent diagrams show complex and, in fact, chaotic dynamics of the considered model. However, it is important to note that the numerical analysis was conducted only for fixed parameter values. For different parameters, the results can be significantly different, and in some cases, the system may have a first integral and be integrable. To fully analyze the integrability of the system~\eqref{eq:vv}, we employ the Morales-Ramis theory~\cite{Morales:99::}. This theory examines the differential Galois group of variational equations derived from the linearization of equations of motion  along a particular solution. The main theorem of this theory states that if the Hamiltonian system is integrable in the Liouville sense, then the identity component of the differential Galois group of variational equations must be Abelian.		
 For this article, we mention only the basic facts of this theory concerning the connection between the integrability of Hamiltonian systems with the differential Galois theory. 
For a detailed description of the Morales-Ramis theory and the Kovacic algorithm, see the papers~\cite{Kovacic:86::, Morales:99::,Morales:00::,mp:13::c}. We also refer to the recent results obtained by Combot and Sanabria in the paper~\cite{Combot:18b::}, where the extension of the Kovacic algorithm to the fourth-order linear differential equations is given. 

	Below we formulate the main theorem of this paper.
	\begin{theorem}
	Let $\mu, \delta$ and $\omega$ are positive parameters. Then, the double spring pendulum system governed by Lagrange function~\eqref{eq:lag_{resc}} is not integrable in the class of functions meromorphic in coordinates and velocities.
	\end{theorem}
	\subsection{Outline of the proof}
	System~\eqref{eq:vv} possesses the following invariant manifold
	\begin{equation}
		\label{eq:invariant}
		\scN=\left\{(\ell,v,\varphi_1,\omega_1,\varphi_2,\omega_2)\in \R^6 \, | \, \varphi_1=\omega_1=0=\varphi_2=\omega_2\right\}.
	\end{equation}
	Restricting the right-hand sides of~\eqref{eq:vv} to $\scN$, we get equations for  the harmonic oscillator
	\begin{equation}
		\label{eq:rs0}
		\dot \ell=v,\qquad \dot v=\delta+\omega-\ell,\qquad \ \dot\varphi_1=\dot \omega_1=\dot\varphi_2= \dot \omega_2=0,
	\end{equation}
	with the energy first integral
	\begin{equation}
		\label{eq:e0}
		E=\frac{v^2}{2}+\frac{\ell^2}{2}-\delta \ell-\omega\ell.
	\end{equation}
	Hence, solving equations~\eqref{eq:rs0} and taking into account~\eqref{eq:e0}, we obtain 
a family of particular solutions $\vvarphi(t)=(\ell(t,E),v(t,E),0,0,0,0)$. 	

	Let $\vX=[L, V,\Phi_1,\Omega_1, \Phi_2, \Omega_2]^T$ denote the variations of $\vx=[\ell, v,\varphi_1,\omega_1,\varphi_2,\omega_2]^T$. Then, the variational equations of system~\eqref{eq:vv} along  $\vvarphi(t)$, are as follows
	\begin{equation}
		\dot \vX=\vJ(t)\cdot \vX,\quad \text{where}\quad \vJ(t)=\frac{\partial \vv}{\partial\vx}(\vvarphi(t)).
	\end{equation}
	Here $\vv$ denotes the right-hand side of system~\eqref{eq:vv}. The explicit form of the non-constant matrix $\vJ$ is given by
	\begin{equation}
		\vJ(t)=\begin{pmatrix}
			0&1&0&0&0&0\\
			-1&0&0&0&0&0\\
			0&0&0&1&0&0\\
			0&0&\frac{\delta-\mu\omega-\ell(t)}{\mu}&0&\frac{-\delta+\ell(t)}{\mu}&0\\
			0&0&0&0&0&1\\
			0&0&\frac{-\delta+\mu\omega+\ell(t)}{\mu\ell(t)}&0&\frac{\delta-\mu\omega-\ell(t)}{\mu\ell(t)}&-\frac{2 v(t)}{\ell(t)}
		\end{pmatrix}.
	\end{equation}
	We notice that variational equations separate into
blocks: normal variational equations for variables $[\Theta_1,\Omega_1,\Theta_2,\Omega_2]^T$ and tangential equations for variables $[L,V]^T$. Because the tangential system is trivially solvable, for further integrability analysis we take the normal part, which has the form
	\begin{equation}
		\label{eq:normalne}
		\begin{pmatrix}
			\dot \Phi_1\\\dot \Omega_1\\ \dot\Phi_2\\ \dot\Omega_2
		\end{pmatrix}=
		\begin{pmatrix}
			0&1&0&0\\
			\frac{\delta-\mu\omega-\ell(t)}{\mu}&0&\frac{-\delta+\ell(t)}{\mu}&0\\
			0&0&0&1\\
			\frac{-\delta+\mu\omega+\ell(t)}{\mu\ell(t)}&0&\frac{\delta-\mu\omega-\ell(t)}{\mu\ell(t)}&-\frac{2 v(t)}{\ell(t)}
		\end{pmatrix}
		\begin{pmatrix}
			\Phi_1\\\Omega_1\\\Phi_2\\\Omega_2
		\end{pmatrix}.
	\end{equation}
	This system can be written as a one fourth-order differential equation 	\begin{equation}
		\begin{split}
			\label{eq:fourt_time}
			& \ddddot \Phi_1+a_3(t)\dddot \Phi_1+a_2(t)\ddot\Phi_1+a_1(t)\dot \Phi_1+a_0(t)\Phi_1=0,
		\end{split}
	\end{equation}
	with time-dependent coefficients
	\begin{equation}
	\label{eq:coeff}
	\begin{split}
	a_3(t)&=\left(\frac{2\delta v}{(\delta-\ell)\ell}\right),\\
	a_2(t)&=\left(1+\frac{1-\delta}{\mu}+\omega-\frac{\delta+2\delta\mu+3\mu\omega }{\mu\ell}+\frac{3\omega}{\ell-\delta}+\frac{\ell}{\mu}
			\right),\\
			a_1(t)&=2\left(\frac{1}{\mu}+\frac{\omega}{\delta-\ell}+\frac{\mu\omega-\delta}{\mu \ell}\right),\\
			a_0(t)&=\omega\left(1+\frac{1}{\mu}+\frac{3\omega}{\ell-\delta}-\frac{\delta+2\delta\mu+3\mu\omega}{\mu\ell}\right).
	\end{split}
	\end{equation}
	\begin{lemma}
		\label{lem:1}
		Let us assume $\omega=0$, then the fourth-order differential equation~\eqref{eq:fourt_time} factorizes, i.e., coincides with the action  $L_1[L_2[L_3\Phi_1(t)]]=0$, where $L_i$ are differential operators defined as
		\begin{equation}
			\label{eq:factory}
			\begin{split}
			L_1&=\frac{\rmd}{\rmd t}+\left(\frac{v}{\ell}+\frac{v}{1+\ell}+\frac{v}{\delta-\ell}\right),\\
			L_2&=\frac{\rmd^2}{\rmd t^2} +\left(\frac{v}{\ell}-\frac{v}{1+\ell}+\frac{v}{\delta-\ell}\right)\frac{\rmd}{\rmd t}-\frac{(1+\mu+\ell(2+\ell))(\delta-\ell)}{\mu(1+\ell)\ell},\\ L_3&=\frac{\rmd}{\rmd t}.
			\end{split}
		\end{equation}
	\end{lemma}
 \begin{proof}
     Explicit computations are straightforward but lengthy,
so we leave this to the interested reader. 
 \end{proof}
 We split our further integrability analysis into two independent cases, i.e.,   when $\omega\neq 0$ and $\omega=0$. 
\subsection*{Case with $\omega\neq 0$}
To begin with, we  change the independent variable in the variational quation~\eqref{eq:fourt_time} by substituting

\begin{equation}
\label{eq:change}
t\to z=\frac{1}{\ell(t)-\delta},\quad \text{with}\quad E=-\frac{1}{2}\delta(\delta+2\omega).
\end{equation}
	Taking into account the transformation rules for derivatives, we obtain an equation of the following form 
	\begin{equation}
		\label{eq:rationalized}
		\begin{split}
			y^{(4)}+ a_3(z)y'''+a_2(z)y''+a_1(z)y'+a_0(z)y=0,
		\end{split}
	\end{equation}
	where  $y=\Phi_1(t(z))$ and $a_i$ are rational coefficients defined as
	\begin{equation}
		\begin{split}
			a_3(z)&=\frac{6}{z}+\frac{3}{z-z_1}+\frac{2}{z-z_2},\\
			a_2(z)&=\frac{3}{16 \omega ^2 z^2 \left(z-z_1\right){}^2}+\frac{117 \delta  z^2+2 \delta  z+75 z+2}{4 \delta  z^3
				\left(z-z_2\right)}+\frac{\mu +\mu   (\delta +32 \delta  z+27)z+2 (\delta +1) z+2}{4 \delta  \mu  \omega  z^3
				\left(z-z_1\right) \left(z-z_2\right)},\\
			a_1(z)&=\frac{63 \delta  z^2+7 \delta  z+30 z+3}{4 \delta  z^4 \left(z-z_2\right)}+\frac{\mu +\mu   (\delta +5 \delta  z+6)z+2 (\delta +1) z+2}{16 \delta  \mu  \omega ^2 z^4
				\left(z-z_1\right){}^2 \left(z-z_2\right)}\\ &+\frac{2 \mu + [18 \mu + (  (4+26 z)\mu+3)\delta +3]z-1}{4 \delta  \mu  \omega  z^4
				\left(z-z_1\right) \left(z-z_2\right)},\\
			a_0(z)&=\frac{3}{4 z^4 \left(z-z_2\right)}+
   \frac{2 \mu +\delta  \mu  z+2}{16 \delta  \mu  \omega ^2 z^5 \left(z-z_2\right)
				\left(z-z_1\right)^2}+\frac{\mu +2 \delta  \mu  z+1}{4 \delta  \mu  \omega  z^5 \left(z-z_1\right)
				\left(z-z_2\right)}.
		\end{split}
	\end{equation}
The non-integrability proof of the system is based on the necessary conditions formulated in Lemma~\ref{lem:5} contained in the appendix. Therefore, we need to first check whether equation~\eqref{eq:rationalized} has a hyperexponential solution. A function $f(z)$ is called hyperexponential if its logarithmic derivative $f'(z)/f(z)$ is a rational function.
We will now proceed to prove the following.
    \begin{lemma}
	\label{l1}
	Equation \eqref{eq:rationalized} does not admit any hyperexponential solution. 
\end{lemma}
\begin{proof}
If $\omega\delta\neq 0$, then equation has
four distinct singular points
	\begin{equation}
		z_0=0,\quad z_1=\frac{1}{2\omega},\quad z_2=-\frac{1}{\delta},\quad z_\infty=\infty.
	\end{equation}
	Singularities $z_1, z_2$ and $z_\infty$ are regular singular points, while point $z_0$ is irregular.  The exponents at these points are as follows
\begin{equation*}
	\label{expy3}
	\begin{split}
	E_1& = \{ 0, 1/2,1,3/2\},  \qquad E_2 = \{0,1,2\}, \qquad	E_{\infty}&= \{ 0, 1/2,2,5/2\}.
	\end{split}
\end{equation*}
Using Maple, we found that the exponential parts of formal solutions are
following 
\begin{equation}
	\scE(z)	\in \{ 1,z^2, z^{1/4} \exp[\pm 2/\sqrt{\mu z} ]\}.
\end{equation}
Looking for a hyperexponential solution only first has to be considered. 
The first one also has to be rejected because in this case the formal solution
is logarithmic, that is 
\begin{equation}
	\widehat{y}(z) =
 c_0(z)+c_1(z)\ln(z), 
\end{equation}
and $c_1(z)$ vanishes identically only if $ \omega(\delta+\mu+1)=0$. Hence, if it
exists, it has the form 
\begin{equation}
	y(z) = P(z)\prod_{i=0}^2 (z-z_i)^{e_i} , \qquad e_i \in E_i, 
\end{equation}
where $E_0=\{0,2\}$. Now, the degree $d$ of the polynomial $P(z)$ must be
$d = -e_0-e_1-e_2 -e_\infty$, where $e_0=2$. Thus, $d\leq -2$  for an arbitrary
$e_i\in E_i$ with $i=1,2, \infty$.  This ends the proof.  
\end{proof}
To check the second assumption of Lemma~\ref{lem:5},  we have to analyze the second exterior power of
equation~\eqref{eq:rationalized}, that is, the corresponding
system~\eqref{eq:7b}. Rewritten as a scalar equation, it has order six, and it
reads 
\begin{equation}
	\label{eq:ex2}
	\sum_{i=0}^6 b_i(z) w^{(i)} = 0, 
\end{equation}
where coefficients $b_i(z)$ are rational functions, but they are rather long, so
we do not write them here explicitly. This equation has five regular singular
points  
\begin{equation}
\label{eq:ww}
z_1=\frac{\omega}{2},\quad z_2=-\frac{1}{\delta},\quad z_{3,4}=\pm \frac{1}{\sqrt{\mu\omega\delta}},\quad z_\infty=\infty.
\end{equation}  The exponents at these points belong
to the respective sets 
\begin{equation}
	\label{eq:expex2}
	\begin{split}
		&E_1 = \{-1/2, 0, 1/2,1,3/2, 5/2\},  \quad E_2 = \{0,1,2,3\} \\
		&E_{3,4}=\{0,1,2,3,4,6\},\quad
	E_{\infty}= \{ 3/2,3, 7/2,4,9/2,11/2\}.
	\end{split}
\end{equation}
The point $z=0$ is irregular. Of the six formal solutions, only two can be taken into
account when looking for hyperexponential solutions. Their exponential parts are
$\scE(z)	\in \{ z^{-1},z \}$. In effect, there is only one possibility for
the exponential solution of equation~\eqref{eq:ex2}, namely $w(z) = 1/z
\sqrt{z-z_1}$, and it is a solution of these equations. 

In the above, we assumed that all the singular points are pairwise different.
It is possible that for certain choices of parameters, their number is smaller. 
For example, if $\delta=\mu\omega$ then $z_2=z_3$. An analysis of all these
cases showed that equation~\eqref{eq:ex2} has always just one hyper-exponential
solution of the form given above.

\subsection*{Case with $\omega=0$}
For $\omega=0$ it is appropriate to make a linear change of the variational variables in~\eqref{eq:normalne}. Namely,
\begin{equation}
\begin{split}
&X_{1}=\Phi_1+\Phi_2,\quad Y_{1}=\Omega_1+\Omega_2,\\ & X_{2}=\Phi_1-\Phi_2,\quad Y_{2}=\Omega_1-\Omega_2.
\end{split}
\end{equation}
The normal variational equations~\eqref{eq:normalne} in variables $(X,U,Y,V)$ take the form
\begin{equation}
    \begin{pmatrix}
        \dot X_{1}\\
        \dot Y_{1}\\
        \dot X_{2}\\
        \dot Y_{2}\\
    \end{pmatrix}=\begin{pmatrix}
        0&1&0&0\\
        0&-\frac{v}{\ell}&\frac{(v-\ell)(\ell-1)}{\mu \ell}&\frac{v}{\ell}\\
        0&0&0&1\\
        0&\frac{v}{\ell}&\frac{(v-\ell)(\ell+1)}{\mu\ell}&-\frac{v}{\ell}
    \end{pmatrix}\begin{pmatrix}
        X_{1}\\ Y_{1}\\ X_{2}\\ Y_{2}
    \end{pmatrix}.
\end{equation}
We can notice that the equations for $(Y_{1},X_{2},Y_{2})$ do not depend on the variable $X_{1}$, that is, the variable $X_{1}$ decouples from the remaining variables. Thus, it is sufficient to investigate the subsystem of variables $(Y_{1},X_{2},Y_{2})$. Moreover, this subsystem has the first integral
\begin{equation}
    F=\frac{1}{2}\left(1+\mu+2\ell+\ell^2\right)Y_{1}+\frac{1}{2}\left(1+\mu-\ell^2\right)Y_{2}-v X_{2}.
\end{equation}
At level $F=0$, we have
\begin{equation}
    Y_{1}=\frac{2vX_{2}-(1+\mu-\ell^2)Y_{2}        }{1+\mu+2\ell+\ell^2}.
\end{equation}
Therefore, we can eliminate the variable $Y_{1}$, and we end up with the system for the variables $(X_{2},Y_{2})$. It can be rewritten as a one-second-order differential equation
\begin{equation}
	\label{eq:war_omega-zero}
\ddot y+a(t) \dot y+b(t)y=0,\qquad y\equiv X_2,
\end{equation}
where
\begin{equation}\begin{split}
	&a(t)=\frac{2(1+\mu+\ell)v}{\ell(1+\mu+2\ell+\ell^2)},\qquad 
	&b(t)=\frac{(1+\ell)(\ell-\delta)}{\mu\ell}+\frac{2(\delta-\ell)^2}{\ell(1+\mu+2\ell+\ell^2)}.
	\end{split}
\end{equation}
 Then, we make the following change of the independent variable $t\to z=\ell(t)$.  Taking into account the transformation rules for the derivatives, we transform~\eqref{eq:war_omega-zero} into the equation with rational coeﬃcients
\begin{equation}
	\label{eq:rational}
	y''+p(z)y'+q(z)y'=0,\qquad \qquad '\equiv \frac{\rmd}{\rmd z}.
\end{equation}
Explicit forms of coefficients $p(z)$ and $q(z)$ are given by
\begin{equation}
\begin{split}
	&p(z)=\frac{2}{z}+\frac{1}{z-\delta}-\frac{2(z+1)}{(z+1)^2+\mu},\qquad q(z)=-\frac{z+1}{\mu z(z-\delta)}-\frac{2}{z((z+1)^2+\mu)}.
	\end{split}
\end{equation}
Now, we make the classical Tschirnhaus change of the dependent variable
\begin{equation}
y(z)=x(z)\exp\left[-\frac{1}{2}\int p(z)\right],
\end{equation}
which converts~\eqref{eq:rational} into its reduced form
\begin{equation}
		\label{eq:reduced}
	x''=r(z)x,
\end{equation}
with rational coefficient
\begin{equation}
\label{eq:rr}
\begin{split}
 r(z)&=-\frac{1}{4(z-z_1)^2}-\frac{3\mu}{(z-z_2)^2(z-z_3)^2}+\frac{(z+1)^3+2\mu(z+1)+\mu^2}{\mu z(z-z_1)(z-z_2)(z-z_3)}.
 \end{split}
\end{equation}
Equation~\eqref{eq:reduced} has five singularities
\begin{equation}
z_0=0,\quad z_1=\delta,\quad z_{2,3}=-1\pm\rmi\sqrt{\mu},\quad z_\infty=\infty.
\end{equation}
The singularity $z_0$ is a pole of the first order, while
$\{z_1,z_2,z_3\}$ are poles of the second order. The degree of infinity is one.
Taking into account the characters of singularities, we deduce that differential
Galois group of reduced equation~\eqref{eq:reduced} cannot be reducible (first
case of the Kovacic algorithm) or finite (the third case of the Kovacic
algorithm). The differential Galois group can be only dihedral (case 2 of the
Kovacic algorithm) or whole $\operatorname{SL}(2, \C)$. To distinguish
between these cases, we apply the Kovacic algorithm.
 \begin{lemma}
 	The differential Galois group of Eq.~\eqref{eq:reduced} is $\operatorname{SL}(2, \C)$.
 \end{lemma}
\begin{proof}
According to the second case of the Kovacic algorithm, for singularities
$\{z_1,z_2,z_3\}$ with order two,  we define sets of exponents
\begin{equation}
E_{i}=\{2,2\pm 2\sqrt{1+4c_{i}}\}\cap \Z,\qquad i=1,2,3,
\end{equation}
where $c_{i}$ are coefficients in the dominant terms of the Laurent series
expansions of $r(z)$ around $z_i$. Because $z_0$ and $z_\infty$ are
singularities of order one, we define $E_0=\{4\}$ and $E_\infty=\{1\}$. Hence,
the explicit forms of the auxiliary sets $E_{i}$  are given by
\begin{equation}
\begin{split}
	E_{0}=\{4\},\quad E_{1}=\{2,2,2\},\quad E_{2}=\{-2,2,6\},\quad  E_{3}=\{-2,2,6\},\quad E_\infty=\{1\}.
	\end{split}
\end{equation}
Next, following the algorithm, we calculate the Cartesian product $E=E_0\times E_1\times E_2\times E_3\times E_\infty$. From the obtained~27-element list,  we look for these permutations $e=(e_0,e_1,e_2,e_3,e_\infty)$ for which \begin{equation}
d(e)=e_\infty-e_0-e_1-e_2-e_3\in \N_{\operatorname{even}}.
\end{equation}
However, direct computations show that none of $e$ gives a non-negative
$d(e)$. Thus, the algorithm has stopped and the second case of the algorithm
does not occur. Therefore, the differential Galois group of Eq.~\eqref{eq:reduced}
is $\operatorname{SL}(2, \C)$ with a non-Abelian identity component. 
\end{proof}
Following Lemma~\ref{lem:5}, we have shown that for $\omega\neq 0$ variational equation~\eqref{eq:fourt_time} does not possess any hyperexponential solution, and its associate external second power has exactly one hyperexponential solution. Therefore, for $\omega\neq 0$ the double spring pendulum is not integrable in a class of functions meromorphic in coordinates and velocities. Moreover, for $\omega=0$, we have shown that the variational equations are reduced to the one second-order differential equation for which the differential Galois group is $\operatorname{SL}(2,\C)$ with non-Abelian identity component. Therefore, we conclude that the double spring pendulum system governed by the Lagrange function~\eqref{eq:lag_{resc}} is not integrable in a class of meromorphic functions in coordinates and velocities for all values of the parameters. This ends the proof.

\section{Summary and conclusions \label{sec:conclusions}}
The complicated and mostly chaotic dynamics of multiple pendulums is well known but still in a great scientific activity. This is because such systems explain many fundamental phenomena and have found applications in engineering, robotics, and synchronization theory. Currently, we can observe an increase in the work on the study of dynamics and chaos in multiple pendulums of variable length. 
Compared to classic pendulums, these pendulums have variable arms lengths, making their analysis a quite challenging task since they usually have many degrees of freedom and their corresponding differential equations are highly non-linear. 
 As a result, the study of these systems may lead to the discovery of new phenomena related to chaos theory, which is particularly interesting from a scientific point of view. It is also worth noting that chaotic and chaotic phenomena can appear in variable-length systems even for very small perturbations of the parameters. This poses a challenge for researchers analyzing the stability of dynamical systems.

 In this paper, we studied the dynamics and integrability of the variable-length pendulum system, such as the double spring pendulum. It is a Hamiltonian system with three degrees of freedom, so its analysis was quite a challenging task. To gain insight into the dynamics of the system, we joined various numerical methods to get the most reliable results. 
 By joining the Lyapunov exponents with the bifurcation diagrams and Poincar\'e sections as one powerful tool, we gave a complete picture of the system dynamics by specifying values of parameters or initial conditions for which motion of the studied model can be hyperchaotic, chaotic, quasi-periodic, and finally periodic, which is completely new in the context of Hamiltonian systems.   Moreover, in the absence of gravitational potential, the system exhibits $S^1$ symmetry, and the presence of an additional first integral was identified using Lyapunov exponents diagrams. We demonstrate the effective utilization of Lyapunov exponents as a potential indicator of first integrals and integrable dynamics.

The detailed analysis suggests the non-integrability of the proposed model, accompanied by analytical proof. The presence of particular solutions allowed us to apply Morales-Ramis theory. The novelty of our work lies in the integrability analysis of the Hamiltonian system with three degrees of freedom, where the variational equation is transformed into a fourth-order differential equation. To effectively analyze the differential Galois group of variational equations, we employed a recently formulated extension of the Kovacic algorithm designed for dimension four.

In summary, our article presents a comprehensive analysis of the dynamics and integrability of the double spring pendulum, offering new insights and methodologies for further research in this field. We employed powerful tools to obtain results that are of significant importance and usefulness in the study of differential equations, vibrations, and synchronization theory. Variable-length pendulums hold importance from theoretical and practical perspectives, rendering them a fascinating research subject in the field of nonlinear dynamics. 
 Although primarily studied in the realm of nonlinear dynamics and classical mechanics, they have potential applications in various fields, including space debris removal. For instance, the idea for capturing and manipulating debris objects in space mentioned in the introduction can be modeled as a version of a double spring pendulum system. The chaotic nature of the double spring pendulum system shown in this paper presents challenges and opportunities for trajectory planning during debris removal operations.    Advanced algorithms for dynamic trajectory optimization could be developed to exploit the system's nonlinear dynamics and chaotic behavior for efficient debris capture and manipulation. 
Thus, the results obtained in this work could provide qualitative and quantitative information about the possible complex dynamics of the tethered satellite system for active debris removal. Moreover, further studies of the double spring pendulum with some additional dissipations could provide important information about the influence of atmospheric disturbance on the dynamics of tethered satellite system. By continuously adapting the system's parameters in response to changing environmental conditions and debris characteristics, optimal trajectories can be calculated to minimize fuel consumption and maximize operational effectiveness.
\section*{Declaration of competing interest}
The authors declare that they have no known competing financial interests or personal relationships that could have appeared to influence the work reported in this paper.
\section*{Acknowledgements}
 For Open
	Access, the authors have applied a CC-BY public
	copyright license to any Author Accepted Manuscript
	(AAM) version arising from this submission.
		\section*{Funding} This research was funded by The
	National Science Center of Poland under Grant No.
	2020/39/D/ST1/01632.
	\section*{Data availability }
	The data that support the findings of this study are available from the corresponding author,  upon reasonable request.

\appendix
\section{Proof of Theorem}
The key point in the application of Morales-Ramis theory is to decide if the differential Galois group of variational equations is virtually Abelian, that is if its identity component is Abelian. 
In the majority of cases when the theory was implemented, the variational equations split into a set of second-order equations, or had, as a subset, a second-order equation. Furthermore, it was almost always possible to convert these second-order equations into equations with rational coefficients.
Therefore, it was possible to use the Kovacic algorithm formulated in \cite{Kovacic:86::}, which was designed to find the closed form of solutions of the considered second-order equation.  More precisely, it decides if the system admits solutions in a field of Liouvillian functions. The algorithm is based
on the complete classification of the algebraic subgroups
of the group $ \mathrm{SL}(2,\C)$.
As a by-product, it determines the
differential Galois group of the equation.

Equations of motion of a Hamiltonian system with $n$ degrees of freedom  can be written as
\begin{equation}
  \dot \vz=\mathbb{J} H'(\vz),\quad \mathbb{J}=\begin{bmatrix} 0&\mathbb{I}_n\\
    -\mathbb{I}_n&0\end{bmatrix},\quad \vz=[\vq,\vp]^T,
\end{equation}
and the corresponding  variational equations along  a particular
solution $\varphi(t)$ are also Hamiltonian
\begin{equation}
  \dot \vY=\mathbb{J}H''(\varphi(t))\vY.
\end{equation}
It is easy to show that the differential Galois group of this system is a subgroup of the symplectic group,
$ \mathrm{Sp}(2n,\C)$. For $n=1$ the group $\mathrm{Sp}(2,\C)$ is isomorphic to
$\mathrm{SL}(2,\C)$. However, for larger dimensions
$\mathrm{Sp}(2m,\C)\subset\mathrm{SL}(2m,\C)$, it is much smaller than $\mathrm{SL}(2m,\C)$. 
The classification of the subgroups of $\mathrm{Sp}(4,\C)$ is known. 
Among other things, this fact was used in \cite{Combot:18::}
where the equivalent of the Kovacic
algorithm for symplectic differential operators of degree four was formulated.

For a brief description of this algorithm, we introduce the appropriate terminology. Let $L$
be a differential operator with coefficients in $\C(z)$
\begin{equation}
  L(y)=y^{(n)}+a_{n-1}y^{(n-1)}+\cdots+a_1y'+a_0y=0,\quad a_i\in\C(z),
\end{equation}
and $A$ is its corresponding companion matrix, that is
\begin{equation}
  A=\begin{bmatrix}
    0&1&0&\cdots&0\\
    0&0&1&\cdots&0\\
    \vdots&\vdots&\vdots&\vdots&\vdots\\
    0&0&0&\cdots&1\\
    -a_0&-a_1&-a_2&\cdots&-a_{n-1}
  \end{bmatrix}.
  \label{eq:compa}
\end{equation}

An operator of even order  $n=2n$ is
\begin{itemize}
\item symplectic if there exists an invertible skew-symmetric matrix
  $W$ with elements in $\C(z)$ which is a solution of the following
  equation
  
  \begin{equation}
    \label{eq:7b}
    A^TW+WA+W'=0,
  \end{equation}

\item projectively symplectic, if there exists an invertible
  skew-symmetric matrix $W$ with elements in $\C(z)$ which is a
  solution of equation
  \begin{equation}
    \label{eq:8b}
    A^TW+WA+W'+\lambda W=0,
  \end{equation}
  for a certain $\lambda\in\C(z)$.
\end{itemize}

An operator $L$ of order $n=2m$ is symplectic (respectively projectively
symplectic) when its Galois group is isomorphic to a subgroup of
symplectic matrices $\mathrm{Sp}(2m,\C)$ (respectively projectively
symplectic matrices $\mathrm{PSp}(2m,\C)$)
\[
  \begin{split}
    &\mathrm{Sp}(2m,\C)=\{M\in\mathbb{M}_{2m}(\C)\ |\ M^T \mathbb{J} M=\mathbb{J}\},\\
    &\mathrm{PSp}(2m,\C)=\{M\in\mathbb{M}_{2m}(\C)\ |\ M^T \mathbb{J}
    M=\lambda \mathbb{J},\ \lambda\in\C^{\ast}\}.
  \end{split}
\]
If $L$ is projectively symplectic, then up to a multiplication of a
hyper exponential function, the operator can be
symplectic. The function $f(z)$ is called hyper-exponential if its
logarithmic derivative $f'(z)/f(z)$ is a rational function.

\begin{lemma}
  \label{lem:1}
  Assume that the system $\dot x=Ax $ is symplectic, that is, there exists an invertible skew-symmetric matrix $W$ with coefficients in $\C(z)$ which is a solution of equation~\eqref{eq:7b}. Then $x(t)$
  is a solution of $\dot x(t)=Ax(t)$ if and only if $ x_{\ast}(t)=Wx(t)$ is a solution of the adjoining equation $\dot{x}_{\ast}(t)=-A^Tx_{\ast}(t)$.
\end{lemma}
This lemma can be proved by direct check. It follows that if the
symplectic operator $L$ has the right factor $L_{1}$, then its adjoining
$L^{\ast}$ has a right factor $\widetilde L_{1}$ of the same order.
Thus, $L$ has the right factor $L_{1}$, then it also has a left factor of
the same degree.

With a system $\dot x=Ax $, we can associate its external second power.
It is a system of the form
\begin{equation}
  \label{eq:21}
  \dot W = A W -W^{T}A^{T}
\end{equation}
where $W$ is an antisymmetric matrix. Thus, equation~\eqref{eq:7b} is an equation for
the external square of a dual system to system $x'=Ax$.

The classification theorem formulated in \cite{Combot:18::} is the
following.
\begin{lemma}
  \label{lem:T}
  A Lie subgroup of $\mathrm{Sp}(4,\C)$ is up to conjugacy generated
  by elements of the form:
  \begin{enumerate}
  \item upper block triangular matrices with diagonal blocks of size
    at most $2\times 2$,
  \item $2\times 2$ diagonal matrices and anti-diagonal matrices
    \[
      \begin{bmatrix}
        \ast&\ast&0&0\\
        \ast&\ast&0&0\\
        0&0&\ast&\ast\\
        0&0&\ast&\ast
      \end{bmatrix},\quad
      \begin{bmatrix}
        0&0&\ast&\ast\\
        0&0&\ast&\ast\\
        \ast&\ast&0&0\\
        \ast&\ast&0&0
      \end{bmatrix},
    \]
  \item full group $\mathrm{Sp}_{4}(\C)$.
  \end{enumerate}
\end{lemma}
This classification is constructed based on the known classification
of Lie subgroups of the larger unimodular group
$\mathrm{SL}(4,\C)\supset\mathrm{Sp}(4,\C)$. Subgroups of the
projective symplectic group are central extensions of these, and so
contain multiples of the identity matrix with non-unit determinant.
However, since these commute with all matrices, the possible structures of
subgroups in the items $1,2$ are unchanged.

The next two lemmas characterize the reducible case.
\begin{lemma}
  \label{lem:2}
  Let us consider the following block diagonal system
  \begin{equation}
    \label{eq:2b}
    \dot x = A x \qquad 
    A= \begin{bmatrix}
      B & C\\
      0 &  D
    \end{bmatrix},
  \end{equation}  
  where $B$, $C$ and $D$ are matrices $2\times 2$ with rational coefficients. Then the equation~\eqref{eq:7b} has the following
particular solution 
\begin{equation} \label{eq:4} W = \rme^{\int r } \begin{bmatrix}
      0 & 0\\
      0 &  J
    \end{bmatrix}, \qquad J= \begin{bmatrix}
      0 & 1\\
      -1& 0
    \end{bmatrix},
  \end{equation}
  where $r$ is a rational function.
\end{lemma}

In our proof of non-integrability, we will use the following
criterion.
\begin{lemma}
  \label{lem:5}
  Assume that the equation\begin{equation}L(y)=y^{(4)}+a_3(z)y'''+a_2(z)y''+a_1(z)y'+a_0(z)y=0,\quad '=\Dz   
    \label{eq:4order}
  \end{equation}
  is projectively symplectic and $A$ is its corresponding companion matrix. If \eqref{eq:4order} does not admit a hyper exponential solution and equation ~\eqref{eq:7b} has exactly one hyper exponential solution, then the differential Galois group of \eqref{eq:4order} contains $\mathrm{Sp}(4,\C)$.
\end{lemma}
Hence, we must know how to check if a given equation has a hyperexponential solution. 

If $L(y) = 0$ is of Fuchsian type, then any hyper exponential solution must be of
the form $P(z)\prod_i(z-z_i)^{e_i}$, where $P(z)\in\C[z]$, $z_i\in \C$ is a
singular point, $e_i$ are exponents at $z_i$. Then, the necessary conditions for
such a solution is given in Lemma 3.1 in \cite{Singer:95::}. As a corollary, we have
the following.
\begin{proposition}
  If a Fuchsian equation of order four has a factor of order one, then
  either $L(y)=0$ or $L^{\ast}(y)=0$ has a solution of the form
  $P(z)\prod_i(z-z_i)^{e_i}$, where $P(z)\in\C[z]$, $z_i\in \C$ are
  singularities, $e_i$ are exponents at $z_i$, and there exists an
  exponent at infinity $e_{\infty}$ such that the sum
  $\sum_ie_i+e_{\infty}$ is a non-positive integer.
  \label{prop:reduc}
\end{proposition}
For an equation  $L(y) = 0$ which is not Fuchsian, the conditions for the
existence of a hyperexponential solution is more complicated.  We consider only a
particular case, assuming that equation $L(y) = 0$ has only one irregular
singular point at $z=0$. Then its hyperexponential solution must be of the form 
\begin{equation}
	\label{heps}
	y(z) = \scE(z) R(z)\prod_{i=1}^m(z-z_i)^{e_i},
\end{equation}
where $R(z)$  is a rational function, $z_i\in \C$ is a
singular point, $e_i$ are exponents at $z_i$ for $i=1,\ldots, m$. The function
$\scE(z)$ is an exponential part of a formal solution at the irregular point
$z=0$  
\begin{equation}
    \widehat{y}(z) =
\scE(z)\left[ c_0(z^{1/k})+c_1(z^{1/k})\ln(z)+\cdots + c_1(z^{1/k})\ln(z)^m\right]
\end{equation}
where 
\begin{equation}
	\label{exppart}
	\scE(z)	= \exp[W(z^{-1/k})]z^a,  
\end{equation} 
$W$ is a polynomial, $k>0$ and $m\geq 0$,  $a\in\C$ and  $c_i$ are formal power
series.   In practice, formal solutions can be found with the help of a computer
algebra system, for example, MAPLE. 

For the equation of order $n$, there are $n$ linearly independent solutions of this
form. For formula \eqref{heps} we take only those for which $k=1$ and $m=0$. If
the exponential part is $\scE(z)	= z^a$, then $a$ is an exponent at a singular
point $z=0$.  In this case, if a hyperexponential exists it has the same form as
for Fuchsian equation.

\end{document}